\documentclass[a4paper,USenglish,cleveref,autoref]{lipics-v2021}

\pdfoutput=1 
\hideLIPIcs  

\usepackage[sort]{cite} 
\usepackage[dvipsnames]{xcolor}
\definecolor{Purple}{rgb}{0.63, 0.13, 0.94}
\definecolor{DarkGreen}{rgb}{0.0, 0.39, 0.0}
\usepackage{marginnote}
\usepackage{comment}

\usepackage{bm}

\usepackage{thmtools}
\usepackage{thm-restate}

\newcommand{\Real}{\ensuremath{\mathbb{R}}}
\newcommand{\Plane}{\ensuremath{\mathbb{R}^2}}

\newcommand{\bd}{\mathop{\partial}} 
\newcommand{\cl}{\mathop{\mathrm{cl}}} 
\newcommand{\intr}{\mathop{\mathrm{int}}} 

\newcommand{\Bisector}[2]{J({#1},{#2})} 
\newcommand{\Bisectors}{\mathcal{J}} 
\newcommand{\Dominance}[2]{D({#1},{#2})} 

\newcommand{\sprec}{<}
\newcommand{\spreceq}{\leq}
\newcommand{\sequiv}{=}
\newcommand{\colprec}{\prec}
\newcommand{\colpreceq}{\preccurlyeq}
\newcommand{\colequiv}{\sim}
\newcommand{\mcolprec}{\mathrel{\overline{\succ}}}
\newcommand{\mcolpreceq}{\mathrel{\overline{\succcurlyeq}}}
\newcommand{\mcolequiv}{\mathrel{\overline{\sim}}}

\newcommand{\VD}{\ensuremath{\mathsf{VD}}} 
\newcommand{\VR}{\ensuremath{\mathrm{VR}}} 
\newcommand{\FVD}{\ensuremath{\mathsf{FVD}}}
\newcommand{\FVR}{\ensuremath{\mathrm{FVR}}}
\newcommand{\CVD}{\ensuremath{\mathsf{CVD}}}
\newcommand{\CVR}{\ensuremath{\mathrm{R}}} 
\newcommand{\FCVD}{\ensuremath{\mathsf{FCVD}}}
\newcommand{\FCVR}{\ensuremath{\mathrm{FR}}}
\newcommand{\mCVD}{\ensuremath{\overline{\mathsf{CVD}}}}
\newcommand{\mCVR}{\ensuremath{\overline{\mathrm{R}}}}
\newcommand{\HVD}{\ensuremath{\mathsf{HVD}}}
\newcommand{\HVR}{\ensuremath{\mathrm{HR}}}

\newcommand{\SumV}{V} 
\newcommand{\SumU}{U} 
\newcommand{\mSumV}{\overline{V}} 
\newcommand{\mSumU}{\overline{U}} 

\newcommand{\conf}{\mathcal{F}}

\newcommand{\khat}{\hat{k}}

\newcommand{\mBisector}[2]{\overline{J}({#1},{#2})} 
\newcommand{\mDominance}[2]{\overline{D}({#1},{#2})} 

\DeclareMathOperator*{\E}{\mathbf{E}}

\let\geq\geqslant
\let\leq\leqslant

\newenvironment{denseitems}{\list{$\bullet$}{\itemsep6pt\parsep-5pt}}{\endlist}

\newcommand{\deleted}[1]{}

\newcommand{\evanthia}[1]{}

\newcommand{\nicolau}[1]{}

\title{Abstract Color Voronoi Diagrams and Circular Sequences of Color Permutations}

\titlerunning{Abstract Color Voronoi Diagrams} 

\author{Sang Won Bae}
{Division of AI Computer Science and Engineering, Kyonggi University, Suwon, Republic of Korea}{swbae@kgu.ac.kr}
{https://orcid.org/0000-0002-8802-4247}
{
}

\author{Nicolau Oliver}
{Faculty of Informatics, Universit\`{a} della Svizzera italiana, Lugano, Switzerland}
{nicolau.oliver.burwitz@usi.ch}
{https://orcid.org/0009-0004-8901-451X}
{}

\author{Evanthia Papadopoulou}
{Faculty of Informatics, Universit\`{a} della Svizzera italiana, Lugano, Switzerland}
{evanthia.papadopoulou@usi.ch}
{https://orcid.org/0000-0003-0144-7384}
{}

\authorrunning{S.W. Bae, N. Oliver, and E. Papadopoulou} 

\Copyright{Sang Won Bae and Nicolau Oliver and Evanthia Papadopoulou} 



\category{} 

\nolinenumbers 

\EventEditors{Philip Bille, Seth Pettie, and Sabine Storandt}
\EventNoEds{3}
\EventLongTitle{34th Annual European Symposium on Algorithms (ESA 2026)}
\EventShortTitle{ESA 2026}
\EventAcronym{ESA}
\EventYear{2026}
\EventDate{August 31--September 4, 2026}
\EventLocation{L'Aquila, Italy}
\EventLogo{}
\SeriesVolume{388}
\ArticleNo{136}

\begin{document}

\maketitle

\begin{abstract}
  Abstract Voronoi diagrams are defined in terms of a given system of planar
  bisecting curves satisfying some simple combinatorial properties. They offer a
  unifying framework for a wide range of concrete Voronoi instances on generalized
  sites and metrics. In this paper, we
  formulate higher-order abstract color Voronoi diagrams of a set $S$ of $n$
  colored abstract sites, simultaneously considering all
  concrete instances under their umbrella.
  We prove that the number of vertices in the order-$k$ abstract
  color Voronoi diagram is at most $4k(n-k)-2n$, and present an iterative
  construction algorithm. The bound directly applies to a family of $m$ disjoint
  simple polygons of total complexity $n$.
  For simple polygons the bound can further
  improve to $O(\min\{k(n-k),(m-k)^2n\})$. A critical ingredient of our proof is
  a combinatorial analysis on circular sequences of color permutations derived
  from the unbounded edges of these diagrams, which is interesting in its own right.
\end{abstract}

\newpage

\section{Introduction} \label{sec:intro}

Voronoi diagrams are versatile and influential space partitioning structures.
Given a set~$S$ of $n$~sites in $\Plane$  and an underlying distance function,
the ordinary Voronoi diagram~$\VD(S)$ partitions the plane into maximal
regions by the nearest site relation.
The order-$k$ Voronoi diagram~$\VD_k(S)$ partitions $\Plane$
into regions by the $k$~nearest sites for~$1\leq k\leq n-1$, where $\VD_1(S)=\VD(S)$.
The farthest-site Voronoi diagram~$\FVD(S)$ is equal to~$\VD_{n-1}(S)$. 
%
Sites may often be assumed to be points in the Euclidean plane, however,
generalized sites, such as disks, line
segments and polygons, under generalized metrics may  also constitute the
input sites.
See 
\cite{akl-vddt-13,obsc-st-00} for
extensive information. 

Lee~\cite{l-knnvdp-82} proved the tight
bound~$O(k(n-k))$ on the combinatorial complexity of~$\VD_k(S)$ of point sites in the Euclidean plane,
and presented an iterative algorithm that computes the diagrams order~$1$ up to~$k$
in $O(k^2n\log n)$ time.
The $O(k(n-k))$ bound on the complexity of~$\VD_k(S)$  has been extended to line segments 
under any~$L_p$ metric~\cite{pz-hovdls-16} and to abstract Voronoi diagrams \cite{bcklpz-choavd-15}.
For point sites in the $L_1/L_\infty$ metric a better bound~$O(\min\{k(n-k), (n-k)^2\})$ is
known~\cite{lpl-knnvdr-15}.
The problem of constructing the Euclidean order-$k$ Voronoi diagram~$\VD_k(S)$ for point sites~$S$
had been one of the most interesting open problems in computational geometry,
and the first optimal $O(n\log n + k(n-k))$-time algorithm
was presented recently by Chan et al.~\cite{ccz-oahovdp:usn-24},
after a series of algorithmic advances for over four
decades~\cite{agss-ltacvdcp-89,as-soriachovd-92,m-lavd-91,adms-clahovd-98,c-rshrrcltd-00,r-rrrsklc-99,ct-oda2d3dsc-16}.
For generalized sites, however, there is a notable scarcity of corresponding  results.

%


\deleted{
Lee and Drysdale~\cite{ld-gvdp-81} analyzed the $\VD(S)$
of disjoint simple polygons (treating them as sets of line segments),
also known as the \emph{medial axis transform}~\cite{b-tends-67,l-matps-82}.
Cheong et al.~\cite{cegghlln-fpvd-11} showed that the complexity of
the $\FVD(S)$ of $m$ disjoint connected polygonal sites of total complexity $n$ is $O(n)$,
and can be computed in $O(n \log^3 n)$ time. However, despite the importance of
polygonal sites in applications, the combinatorial complexity of the $\VD_k(S)$
of polygons has remained a significant open problem.
}

In \emph{color Voronoi diagrams}, 
colors are  assigned to the
sites in $S$, modeling a common property that sites of the same
color share; let~$K$ be the set of these $m \leq n$~colors.
The color assignment aggregates simple sites, such as points, segments, or disks,  into
compound ones of non-constant complexity, such as simple polygons, arc
polygons, and site clusters.
The color, a non-spatial property, facilitates the modeling of diverse 
applications, including facility 
location~\cite{ahiklmps-fcvdrp-06}, shape matching~\cite{hks-uevsia-93}, spatial
databases~\cite{chen2020}, wireless sensor networks~\cite{lsbc-bwccpapwsn-13},
nearest-neighbor classification~\cite{bremner2005},
fault detection and analysis in VLSI
networks~\cite{P11} and references therein. 
Nearest, higher-order, and farthest color Voronoi diagrams can be naturally
defined. 
%

Different variants of color Voronoi diagrams have been 
considered in the literature, such as the \emph{Hausdorff} (also called
\emph{cluster}) 
\emph{Voronoi diagram}, 
e.g.,~\cite{egs-ueplf:aa-89,p-cacmmdVLSIc-01,p-hvdpcp-04,P11,AP19},
and the \emph{farthest color Voronoi diagram},
e.g.,~\cite{hks-uevsia-93,ahiklmps-fcvdrp-06,b-tbiafcvdls-14, mpsw-fcvd:ca-25},
as motivated by different application demands, see e.g., \cite{VLSI-bookchapter,caa}.
\emph{Higher-order
  color Voronoi diagrams} 
were recently introduced  in \cite{bop-hocvdccsf-25}. 
In $\Real^d$, \emph{chromatic Delaunay mosaics} of colored point sets
have  also recently been introduced as motivated by topological data analysis~\cite{bmdes-scdm-26}.

Concrete color Voronoi diagrams are based on distance-to-color functions \cite{bop-hocvdccsf-25}: 
for each color~$a \in K$ and any point~$x\in \Plane$, let
$d_a(x) := \min_{s\in S_a} \delta_s(x)$ and $\bar{d}_a(x) := \max_{s_\in S_a} \delta_s(x)$
be the minimal and maximal distance-to-color~$a$ from~$x$,
where $\delta_s(x)$ denotes the prescribed distance to site~$s$ from~$x$, and
$S_a\subseteq S$ is the set of sites of color $a$.
For each~$1\leq k\leq m-1$,
the \emph{order-$k$ minimal color Voronoi diagram}~$\CVD_k(S)$ of colored sites~$S$
partitions~$\Plane$ into regions by $k$~nearest colors with respect to
the minimal distance-to-color functions~$\{d_a\}_{a\in K}$,
while the \emph{order-$k$ maximal color Voronoi diagram}~$\mCVD_k(S)$
partitions~$\Plane$ by $k$ farthest colors with respect to
the maximal distance-to-color functions~$\{\bar{d}_a\}_{a\in K}$.
In~\cite{bop-hocvdccsf-25} the authors
proved the tight upper bound~$4k(n-k)-2n$
on the total number of vertices in~$\CVD_k(S)$ and $\mCVD_k(S)$
for a range of well-behaved distance to site
functions, 
under a set of assumptions 
satisfied by  
point sites under convex distance functions, and presented an iterative construction algorithm.
The assumptions of~\cite{bop-hocvdccsf-25}, however, are not satisfied by
non-point sites, thus, the derived bounds do not apply in any 
setting involving sites more general than points.
In this paper we remove these assumptions generalizing upon~\cite{bop-hocvdccsf-25}.

The order-$k$ color Voronoi diagrams~$\CVD_k(S)$ and $\mCVD_k(S)$
build upon 
several known diagrams that have 
received extensive interest in the literature.
For $k=1$, the minimal diagram~$\CVD_1(S)$ is a subset of $\VD(S)$, while the
maximal diagram~$\mCVD_1(S)$ is contained in  $\FVD(S)$.
If $m=n$, that is, every site in~$S$ has a distinct color,
then $\CVD_k(S)=\VD_k(S)$ and $\mCVD_k(S) = \VD_{n-k}(S)$.
For $k=m-1$, $\CVD_{m{-}1}(S)$ is exactly the  \emph{farthest color
  Voronoi diagram}~$\FCVD(S)$, 
and $\mCVD_{m-1}(S)$ is the \emph{Hausdorff Voronoi diagram}~$\HVD(S)$.

\emph{Abstract Voronoi diagrams} (AVDs)~\cite{k-cavd-89} 
offer a
unifying framework for a wide range of concrete Voronoi instances.
Rather than 
sites and distance measures, AVDs are defined
in terms of bisecting curves that satisfy some simple combinatorial properties. 
Examples of concrete diagrams under the AVD umbrella include Voronoi diagrams of (non-intersecting) line segments, or
convex polygons of constant complexity, in the $L_p$ norms;
Euclidean Voronoi diagrams of (non-intersecting) disks  or
smooth convex objects;  point sites in any convex distance metric;
additively weighted points, and  power diagrams, 
with non-enclosing circles
in \emph{nice} metrics. 
Bohler et al.~\cite{bcklpz-choavd-15} proved 
the tight upper bound~$2k(n-k)+k+1-n$ on the number of faces in
the order-$k$ abstract Voronoi diagram~$\VD_k(S)$.
Efficient construction algorithms 
are also known~\cite{blpz-rdcahoavd-16,bkl-erahoavd-19}.

\subparagraph*{Contribution.}
In this paper, we formulate higher-order color Voronoi diagrams
under the AVD model, 
namely \emph{higher-order abstract color Voronoi diagrams}.
By using the abstract framework, our combinatorial results are simultaneously
applicable to all the concrete instances that fall under the AVD
umbrella, which are also interesting individually. 
The obtained upper bounds are tight and analogous to those in
\cite{bop-hocvdccsf-25}, however, they are applicable to a much wider class of problems.
This includes the higher-order Voronoi diagram of $m$ simple polygons whose structural
complexity,  other than its farthest-site counterpart~\cite{cegghlln-fpvd-11},
had been unknown. 
It also results in the first study on the Hausdorff and  farthest color Voronoi diagram  in the
abstract setting, which encompass all
the previously considered concrete cases.

We prove tight upper bounds~$O(k(n-k))$ on the combinatorial complexity of
order-$k$ abstract color Voronoi diagrams~$\CVD_k(S)$ and $\mCVD_k(S)$ for each~$1\leq k\leq m-1$.
More precisely,
we show that the number of vertices of the abstract~$\CVD_k(S)$ and
$\mCVD_k(S)$
is at most $4k(n-k)-2n$, which is tight. 
This immediately implies the worst-case  bound $O(m(n-m+1))$
on the complexity of the abstract farthest color Voronoi diagram~$\FCVD(S)$
and the abstract Hausdorff Voronoi diagram~$\HVD(S)$, which applies to all the
concrete cases under the AVD umbrella.
If $m$ is close to $n$, then our new bound improves the previous one $(O(mn))$ 
   on the complexity of the farthest color Voronoi diagram of
non-crossing line segments in the $L_p$ metric~\cite{b-tbiafcvdls-14,hks-uevsia-93}.
 For a family of $m$ disjoint simple polygons of total complexity $n$, we show
 that  the order-$k$ polygon Voronoi diagram has  complexity
 $O(\min\{k(n-k),(m-k)^2n\})$. This is derived from a more general result:  if the farthest color Voronoi
 diagram  $\FCVD(S')$ (resp. $\HVD(S')$) for any subset~$S'\subseteq S$
 has linear complexity,
 then $\CVD_{k}(S)$ (resp. $\mCVD_{k}(S)$) has complexity
 $O((m-k)^2 n)$ for  $\frac{2m}{3} \leq k\leq m-1$.
 Note that the farthest polygon Voronoi diagram has linear
 complexity~\cite{cegghlln-fpvd-11}.
 The complexity of the order-$k$ Voronoi diagram of polygons 
 had been an open problem for a long time.

 To obtain these results we  first adapt the \emph{colorful Clarkson--Shor
   technique} from~\cite{bop-hocvdccsf-25,cs-arscgII-89} to abstract Voronoi
 diagrams, expressing the combinatorial
complexity of the order-$k$  abstract color Voronoi diagram as a function
of the diagram's unbounded edges (see Lemma~\ref{lem:complexity}).
This reduces the problem into counting the unbounded edges
of the colored order-$k$ diagrams.
For $m=n$ (uncolored case) this already answers a question posed by the authors
of~\cite{bcklpz-choavd-15} on whether  the Clarkson--Shor
technique~\cite{cs-arscgII-89} could be applied in deriving the complexity of
the abstract order-$k$ diagram. 
Deriving bounds on the unbounded edges of colored order-$k$ diagrams is a
critical part of our proof, which is a
novel contribution of this paper. 
In \cite{bop-hocvdccsf-25} the involved quantities were assumed equal and these
bounds could thus be bypassed, which is valid for point sites in convex metrics,
but does not extend to  more general cases.

\deleted{
 obtaining Lemma~\ref{lem:complexity}, which  expresses the combinatorial
complexity of order-$k$  abstract color Voronoi diagrams (and all the concrete cases under
their umbrella) in terms of  functions involving  the diagram's unbounded edges.
For the  uncolored case ($m=n$) this already settles a question posed by the authors
of~\cite{bcklpz-choavd-15} on how to use  the original Clarkson--Shor
   technique~\cite{cs-arscgII-89} as stated above.
Furthermore, it  reduces the problem into to counting the unbounded edges
of the colored diagrams. 
Deriving these bounds is the 2nd critical ingredient of our proof, which is a
novel contribution of this paper described in the following.
}

 An independently interesting ingredient of our combinatorial results is
 based on a circular sequence of permutations of~\emph{colored} elements,
 which correspond to the unbounded edges of the diagrams.
 Circular sequences of permutations, sometimes called \emph{allowable sequences},
 have been useful as a purely combinatorial tool to analyze 
 several basic geometric structures 
 such as $k$-sets and order types~\cite{ag-nssfspp-86,gp-nkssnpp-84,gp-asotdcg-93}.
 We consider a colored variant of circular sequences,
 in which each element of permutations is assigned a color,
 and prove tight lower and upper bounds on the number of switches, which reflect
 the unbounded edges of order-$k$ abstract color Voronoi diagrams.

 Finally, we show how to construct both~$\CVD_k(S)$ and $\mCVD_k(S)$ in an iterative
 fashion, for  orders~$1$ to~$k$, in time $O(k^2n\log n)$.
 %
 Because $\mCVD_{k-1}(S)$ does not contain sufficient information to compute
 $\mCVD_k(S)$, we first provide a  direct divide-and-conquer
 algorithm to compute the unbounded edges of $\mCVD_{k}(S)$ in time
 $O(k^2(n-k) \log m + n \log n)$,  after which
 $\mCVD_k(S)$ can be computed from $\mCVD_{k-1}(S)$.
 A byproduct of this technique leads to an 
 $O((n-k)^2n\log n)$-time iterative algorithm to
 compute the (uncolored) order-$k$ abstract Voronoi diagram
 $\VD_k(S)$, in decreasing
 order of $k$, starting from $\FVD(S)$, which is efficient for large values of $k$.
 To the best of our knowledge, no such algorithm was known before.
\section{Review on higher-order abstract Voronoi diagrams}
\label{sec:review}

For any $A \subseteq \Plane$,
we denote by $\bd A$, $\intr A$, and $\cl A$
the boundary, interior, and closure of~$A$, respectively.

Let $S$ be a set of $n$ \emph{abstract sites} that define a family 
$\Bisectors = \{\Bisector{p}{q} \mid p,q\in S, p\neq q\}$ of bisecting curves.
The bisector $\Bisector{p}{q}$ of two sites $p,q\in S$ is an unbounded simple
curve, homeomorphic to a line, that divides the plane into two open domains:
the \emph{dominance region} of~$p$ over~$q$ and the dominance region of~$q$ over~$p$, 
denoted by~$\Dominance{p}{q}$ and $\Dominance{q}{p}$, respectively, see \cite{k-cavd-89}.
The dominance region~$\Dominance{p}{q}$ can be regarded as
the set of all points that are \emph{nearer} to~$p$ than to~$q$
according to a relevant proximity notion.

The \emph{nearest} and the \emph{farthest Voronoi region} of site~$p\in S$ are
respectively defined as 
\(
    \VR(p, S) := \bigcap_{q \in S \setminus \{p\}} \Dominance{p}{q}
    \quad\text{and}\quad
    \FVR(p, S) := \bigcap_{q\in S\setminus\{p\}} \Dominance{q}{p}.
\)
The \emph{nearest Voronoi diagram}~$\VD(S)$ and the \emph{farthest Voronoi diagram}~$\FVD(S)$ 
are  
the collection  of the region boundaries:
\(
	\VD(S) := \bigcup_{p\in S} \bd \VR(p, S)
	\quad\text{and}\quad
	\FVD(S) := \bigcup_{p\in S} \bd\FVR(p,S).
\)

The family of bisecting curves 
$\Bisectors$ is called \emph{admissible} if 
it satisfies the following axioms~\cite{k-cavd-89},  for every~$S'\subseteq S$:
\begin{itemize} 
\item [(A1)] Each nearest Voronoi region $\VR(p, S')$ is non-empty and pathwise connected.
\item [(A2)] Each point of the plane belongs to the closure of a nearest Voronoi region $\VR(p, S')$.
\item [(A3)] Each bisector $\Bisector{p}{q}$ is unbounded. After stereographic
  projection to the sphere, it can be completed to a closed Jordan curve through the north pole.
\item [(A4)] Any two bisectors intersect transversally, in a finite number of
  points. 
\end{itemize}


Axiom (A4) can be relaxed as shown in~\cite{kln-avdr-09} but for technical
simplicity we still require  it.
The verification of the axioms can be done with constant size ($\leq 4$)
examples, see \cite{bcklpz-choavd-15}.
Bisectors that have a site in common are called \emph{related}.
When two related bisectors $\Bisector{p}{q}$ and $\Bisector{p}{r}$ intersect, bisector
$\Bisector{q}{r}$ intersects with them at the same point(s).
If $\Bisectors$ is admissible, 
there are at most two such intersection points, which are the vertices of $\VD(\{p,q,r\})$.
%
%
Then $\VD(S)$ and $\FVD(S)$ are plane graphs of complexity
$O(n)$~\cite{k-cavd-89}, and $\FVD(S)$ is a tree~\cite{mmr-fsavd-01}.

The \emph{higher-order abstract Voronoi diagram} was first introduced in~\cite{bcklpz-choavd-15}.
For~$1\leq k\leq n-1$ and each subset~$H \subset S$ with cardinality~$k$,
the \emph{order-$k$ Voronoi region} of~$H$ 
and the \emph{order-$k$ abstract Voronoi diagram}~$\VD_k(S)$ are defined as
\[
    \VR_k(H,S) := \bigcap_{p \in H, q \in S \setminus H} \Dominance{p}{q}
    \quad\text{and}\quad
	\VD_k(S) := \bigcup_{H\subseteq S, |H|=k} \bd \VR_k(H, S).
\]
Note that
$\VD_1(S) = \VD(S)$ and $\VD_{n-1}(S) = \FVD(S)$.
%
The higher-order abstract Voronoi diagram~$\VD_k(S)$ is a plane graph  
that encompasses classic concrete cases, including 
points~\cite{l-knnvdp-82} and segments~\cite{pz-hovdls-16} in the Euclidean plane, 
from which the tight combinatorial complexity bound~$O(k(n-k))$ is also
derived~\cite{bcklpz-choavd-15}.

The dominance regions define an abstract notion of \emph{$k$-th nearest site}.
In particular, given $\Bisectors$ and a 
point $x\in\Plane$, a total order of $S$
can be defined~\cite{bcklpz-choavd-15,blpz-rdcahoavd-16}.  
For $p,q\in S$, we write
\[
p \sprec_x q~\text{iff}~p\neq q~\text{and}~x \in \Dominance{p}{q}, \;
p \sequiv_x q~\text{iff}~p = q~\text{or}~x \in \Bisector{p}{q}, \;
p \spreceq_x q~\text{iff}~p\sprec_x q~\text{or}~p \sequiv_x q.
\]
These relations can be interpreted as
$x$ being \emph{nearer} to~$p$ than to~$q$,
\emph{equidistant} to~$p$ and~$q$, and
\emph{not farther} from~$p$ than from~$q$, respectively.
The axioms imply
the transitivity of~$\sprec_x$~\cite{k-cavd-89} and its reflexive closure
$\spreceq_x$ \cite[Lemma~2]{blpz-rdcahoavd-16}, which 
define a total order on $S$ allowing us to say that $p_k$ is the
\emph{$k$-th nearest site at~$x$}.
In this sense, the order-$k$ region~$\VR_k(H, S)$ of~$H\subseteq S$
consists of all points~$x$ such that $H$ is the set of
$k$~nearest sites at~$x$.

Let $\Gamma$ be a closed simple curve in~$\Plane$ sufficiently large
such that no pair of bisecting curves in~$\Bisectors$ cross on or outside of~$\Gamma$;
furthermore, every bisecting curve intersects~$\Gamma$  transversally, exactly twice.
By this construction,  $\Gamma$ encloses every vertex of~$\VD_k(S)$, for $1\leq k\leq n-1$,
and the unbounded features of~$\VD_k(S)$ are traversed by~$\Gamma$ in the
same order as they appear at infinity.
We consider~$\Gamma$ as  \emph{the closed curve at infinity}.

An admissible  bisector system~$\Bisectors$ 
is said to be in \emph{general position}, if only  three related bisecting curves
can intersect at the same  point.
From now on, we assume that the given system~$\Bisectors$ of bisecting curves is admissible and in general position.

\subparagraph*{Circular sequences of permutations.}
\label{subsec:perm}
The authors of~\cite{bcklpz-choavd-15} proved that
the number of faces in the order-$k$ diagram~$\VD_k(S)$ is at most $2k(n{-}k){+}k{+}1{-}n$ and at least $n-k+1$.
Their combinatorial result is based on an inductive approach
extending the original method by Lee~\cite{l-knnvdp-82}, which had been extended to
segment sites in~\cite{pz-hovdls-16},
and a careful analysis on the circular sequences induced by the unbounded edges of~$\VD_k(S)$, 
for all orders~$1\leq k \leq n-1$, on the closed curve $\Gamma$ at infinity.

For each~$x\in \Gamma$, consider the permutation~$\pi=(p_1, \ldots, p_n)$ of~$S$
induced by~$\spreceq_x$.
If $x$ avoids all bisecting curves in~$\Bisectors$,
then $\spreceq_x$ induces a unique permutation~$(p_1, \ldots, p_n)$ 
such that $p_1 \sprec_x \cdots \sprec_x p_n$;
otherwise, there is a unique index~$j$ such that $x \in \Bisector{p_j}{p_{j+1}}$ and hence
\(
  p_1\sprec_x \cdots \sprec_x p_j \sequiv_x p_{j+1} \sprec_x \cdots \sprec_x p_n.
\)
The corresponding permutation of~$S$ changes by a \emph{switch} of two consecutive sites
when and only when we cross a bisector~$\Bisector{p}{q}$, while we walk along~$\Gamma$. 
As a result, we obtain a circular sequence~$\Pi(S) = (\pi_0, \ldots, \pi_{N-1}, \pi_N = \pi_0)$ 
of permutations~$\pi_i$ of sites~$S$  such that: 
\deleted{
\hypertarget{def:permP1}{\textbf{(P1)}} $\pi_i$ and $\pi_{i+1}$ for any~$i$ differ by a switch of two consecutive elements; and 
 \hypertarget{def:permP2}{\textbf{(P2)}} each pair of two elements switches exactly
 twice.
 }
 \begin{itemize} 
  \item \hypertarget{def:permP1}{(P1)} $\pi_i$ and $\pi_{i+1}$ for any~$i$ differ by a switch of two consecutive elements; and 
  \item \hypertarget{def:permP2}{(P2)} each pair of two elements switches
    exactly twice.  ($N=2\binom{n}{2}=n(n-1)$).
  \end{itemize}
  
%
In~\cite{bcklpz-choavd-15} the authors showed that
each switch between positions~$k$ and~$k+1$ in $\pi_i$ and $\pi_{i+1}$
corresponds to an unbounded edge of 
$\VD_k(S)$.
Thus, they obtained tight lower and upper bounds
on the total number of unbounded edges in the diagrams~$\VD_1(S), \ldots, \VD_k(S)$
by proving tight bounds on the number of switches in the first $k$ positions: 
at least $k(k+1)$ and at most $k(2n-k-1)$.

\section{Defining abstract color Voronoi diagrams}
\label{sec:def_CVDk}

Let $S$ be a set of $n$~sites and $\Bisectors = \{\Bisector{p}{q} \mid p,q\in S, p\neq q\}$
be an admissible bisector system 
in general position.
We assume that each site~$p\in S$ is assigned a color
from a set~$K = \{1, \ldots, m\}$ of $m\leq n$~colors.
Let $S_a \subseteq S$ be the set of sites of color~$a \in K$.

Abstract color Voronoi diagrams can be seen as Voronoi diagrams of 
\emph{colors} in~$K$, where each color $a\in K$ represents its entire color
class $S_a$.
We consider two kinds of color dominance regions:
for each pair of two distinct colors~$a,b\in K$,
define
\[
\Dominance{a}{b} := \intr \left( \bigcup_{p \in S_a} \cl \VR(p, S_a\cup S_b) \right)
\quad\text{and}\quad
\mDominance{a}{b} := \intr \left( \bigcup_{p \in S_a} \cl \FVR(p, S_a\cup S_b) \right)
\]
called the \emph{minimal} and \emph{maximal color dominance regions}
of color~$a$ over color~$b$, respectively.
See the red and blue regions in Figure~\ref{fig:dominance}.

\deleted{
The minimal color dominance region
$\Dominance{a}{b}$ contains all points that are closer to some site of color $a$
than to any site of color $b$. 
Symmetrically, the maximal color dominance region $\mDominance{a}{b}$ contains
all points that are 
farther from some site of color $a$ than from any site of color $b$. 
}

Let $\Bisector{a}{b} := \bd \Dominance{a}{b}$ and $\mBisector{a}{b} :=
\bd \mDominance{a}{b}$
be the \emph{minimal} and \emph{maximal color bisector}, respectively, between
two colors~$a,b\in K$. 
Since $\Bisector{a}{b}$ (resp. $\mBisector{a}{b}$) separates  Voronoi regions
of different color in $\VD(S_a\cup S_b)$ (resp. $\FVD(S_a\cup S_b)$), it may consist of several
disjoint simple curves, either unbounded or closed; see Figure~\ref{fig:dominance}.
In particular, the minimal bisector $\Bisector{a}{b}$ consists
of one or more unbounded or closed curves,
while the maximal bisector~$\mBisector{a}{b}$ consists of only unbounded curves
and can also be an empty set.
These color bisectors may contain degree-$2$ vertices.

\deleted{
  Since $\Bisector{a}{b}$ (resp. $\mBisector{a}{b})$ consists of
 Voronoi edges of $\VD(S_a\cup S_b)$ (resp. $\FVD(S_a \cup S_b)$)
 between two regions of different colors,
it consists of one or more disjoint simple curves,
 either unbounded or closed; see \figurename~\ref{fig:dominance}.
}


\begin{figure}[ht]
    \centering
    \includegraphics[height=4.6cm]{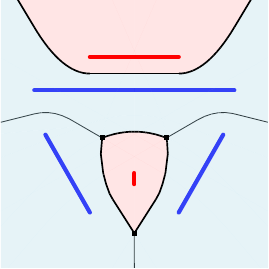}
    \hspace{0.8cm}
    \includegraphics[height=4.6cm]{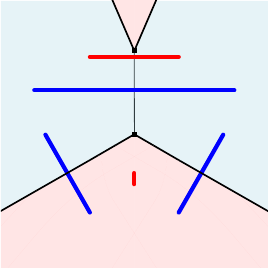}
    \caption{%
        Minimal 
        and maximal 
        dominance regions.
        The red segments are $S_a$
        and the blue ones are $S_b$.
        The thin edges belong to $\VD(S_a\cup S_b)$ and $\FVD(S_a\cup S_b)$, 
        the color bisectors are in bold.
    }
    \label{fig:dominance}
\end{figure}

The two kinds of color dominance regions 
define two hierarchies of higher-order abstract color Voronoi diagrams; see
Figures~\ref{fig:cvd}(a--c) and~\ref{fig:mcvd}(a--c). 
For each subset~$H\subseteq K$ of $k$~colors,
the \emph{order-$k$ minimal} and \emph{maximal color Voronoi regions}
are defined as 
\[
	\CVR_k(H,S) := \bigcap_{a \in H, b \in K \setminus H} \Dominance{a}{b}
	\quad\quad
	\text{and}
	\quad\quad
	\mCVR_k(H,S) := \bigcap_{a \in H, b \in K \setminus H} \mDominance{a}{b}.
\]
The \emph{order-$k$ minimal color Voronoi diagram}~$\CVD_k(S)$ and
the \emph{order-$k$ maximal color Voronoi diagram}~$\mCVD_k(S)$ are
defined as 
\[
 \CVD_k(S) := \bigcup_{H\subseteq K, |H|=k} \bd \CVR_k(H,S)
 \quad\text{and}\quad
 \mCVD_k(S) := \bigcup_{H\subseteq K, |H|=k} \bd \mCVR_k(H,S).
\]

We consider the following color dominance relations on colors~$K$ at any point~$x\in\Plane$,
analogously to~$\sprec_x$, $\sequiv_x$, and $\spreceq_x$ on sites~$S$, 
defined in~\cite{bcklpz-choavd-15}. 
For any colors $a,b\in K$, we write
\[
\begin{aligned}
&a \colprec_x b~\text{iff}~a\neq b~\text{and}~x\in \Dominance{a}{b}, \;
a \colequiv_x b~\text{iff}~a = b~\text{or}~x\in \Bisector{a}{b}, \;
a \colpreceq_x b~\text{iff}~a \colprec_x b~\text{or}~a \colequiv_x b;\\
&a \mcolprec_x b~\text{iff}~a\neq b~\text{and}~x \in \mDominance{a}{b}, \;
a \mcolequiv_x b~\text{iff}~a = b~\text{or}~x \in \mBisector{a}{b}, \;
a \mcolpreceq_x b~\text{iff}~a \mcolprec_x b~\text{or}~a \mcolequiv_x b.
\end{aligned}
\]

\begin{restatable}{lemma}{transitivity} \label{lem:transitivity}
For any~$x\in\Plane$, both $\colpreceq_x$ and $\mcolpreceq_x$ are transitive.
\end{restatable}
\begin{proof}
    It is enough to show that for colors $a,b,c \in K$
    \[
    \cl \Dominance{a}{b} \cap \cl \Dominance{b}{c} \subseteq \cl \Dominance{a}{c}
    \quad\text{and}\quad
    \cl \mDominance{a}{b} \cap \cl \mDominance{b}{c} \subseteq \cl \mDominance{a}{c}.
    \]
    Let $x \in \cl \Dominance{a}{b} \cap \cl \Dominance{b}{c}$. Then $x\in
    \cl \VR(p_a, S_a\cup S_b) $ and $x\in \cl \VR(p_b, S_b\cup S_c)$, for
    some site $p_a \in S_a$ and $p_b \in S_b$.
    It follows that $x\in \cl \VR(p_a, S_a\cup S_b\cup S_c) $ by the standard
    transitivity property of ordinary Voronoi regions. Then $x\in
    \cl \VR(p_a, S_a\cup S_c)$, and consequently, $x\in \cl \Dominance{a}{c}$.
    The transitivity of~$\mcolpreceq_x$ based on the maximal dominance 
    can be proved analogously using the diagrams $\FVD(S_a\cup S_b)$ and
    $\FVD(S_b\cup S_c)$. 
\end{proof}

Lemma~\ref{lem:transitivity} implies (since both~$\colpreceq_x$
and~$\mcolpreceq_x$ are reflexive) that
$\colpreceq_x$ and $\mcolpreceq_x$ induce two orderings of the
$m$~colors in~$K$ at any point $x\in\Plane$:
$a_1 \colpreceq_x \cdots \colpreceq_x a_m$ and $\bar{a}_1 \mcolpreceq_x \cdots \mcolpreceq_x \bar{a}_m$.
We say that $a_k$ is the \emph{$k$-th nearest color} at~$x$ 
with respect to~$\colpreceq_x$, 
and that $\bar{a}_k$ is the \emph{$k$-th farthest color} at~$x$ 
with respect to~$\mcolpreceq_x$. 
Based on transitivity, we establish the following in Lemmas~\ref{lem:cvrk_covers_plane}, \ref{lem:basics_edge} and \ref{lem:neighboring-regions}: 
(i) the closures of the
minimal (resp. maximal) order-$k$ regions $\CVR_k(H,S)$ (resp. $\mCVR_k(H,S)$)
cover the plane; 
(ii) any two distinct minimal (resp. maximal) order-$k$
regions are disjoint; 
and
(iii) adjacent minimal (resp. maximal) order-$k$ regions differ by exactly one color
that define their common boundary. 
%
\begin{restatable}{lemma}{cvrkcoversplane} \label{lem:cvrk_covers_plane}
    For each $1 \leq k \leq m-1$,
    \[
    \Plane
    = \bigcup_{\substack{H \subseteq K, |H| = k}} \cl \CVR_k(H, S)
    \quad\text{and}\quad
    \Plane
    = \bigcup_{\substack{H \subseteq K, |H| = k}} \cl \mCVR_k(H, S).
    \]
\end{restatable}
\begin{proof}
    As already stated, Lemma~\ref{lem:transitivity} implies that $\colpreceq_x$ and $\mcolpreceq_x$ 
    induce two orderings of the $m$~colors in~$K$ at any point $x\in\Plane$:
    $a_1 \colpreceq_x \cdots \colpreceq_x a_m$ and $\bar{a}_1 \mcolpreceq_x \cdots \mcolpreceq_x \bar{a}_m$.
    Then $x\in\cl \CVR_k(H_x, S)$, where $H_x=\{a_1, \ldots, a_k\}$, and
    $x\in\cl \mCVR_k(\bar{H}_x, S)$, where $\bar{H}_x=\{\bar{a}_1, \ldots, \bar{a}_k\}$. 
    Since every point $x \in \Plane$ is contained in at least one such region, 
    the union of their closures covers the plane.
\end{proof}

\begin{restatable}{lemma}{basicedge} \label{lem:basics_edge}
    For each~$1\leq k\leq m-1$,
    \[
    \CVD_k(S) 
    = \bigcup_{\makebox[20pt][c]{$\substack{H\neq H' \subseteq K\\ |H| = |H'| = k}$}} \cl \CVR_k(H, S) \cap \cl \CVR_k(H',S)
    \;\;\text{and}\;\;
    \mCVD_k(S) 
    = \bigcup_{\makebox[20pt][c]{$\substack{H\neq H' \subseteq K\\ |H| = |H'| = k}$}} \cl \mCVR_k(H, S) \cap \cl \mCVR_k(H',S).
    \]
\end{restatable}
\begin{proof}
    It follows from Lemma~\ref{lem:transitivity} that 
    each point~$x$ in 
    $\CVR_k(H, S)$ (resp., $\mCVR_k(H,S)$)
    has a unique set of $k$~nearest (resp., farthest) colors.
    Thus, 
    two distinct regions~$\CVR_k(H, S)$ and~$\CVR_k(H', S)$
    (resp., $\mCVR_k(H,S)$ and~$\mCVR_k(H', S)$) must have empty intersection.
    Since these regions cover the plane by Lemma~\ref{lem:cvrk_covers_plane},
    their boundary, and thus $\CVD_{k}$ (resp., $\mCVD_{k}$) is exactly the union of intersections of their closures.
\end{proof}

\begin{restatable}{lemma}{neighboringregions}
    \label{lem:neighboring-regions}
    Let $\CVR_k(H, S)$ and $\CVR_k(H', S)$ be two adjacent regions in $\CVD_k(S)$
    with common boundary $E := \cl \CVR_k(H, S) \cap \cl \CVR_k(H', S) \neq \emptyset$.
    There are exactly two colors $a\in H$ and $a'\in H'$, $a\neq a'$,
    such that $H\setminus\{a\} =H'\setminus \{a'\}$ and $E\subseteq
    \Bisector{a}{a'}$. Color $a$ (resp. $a'$) is the \emph{$k$-th nearest
        color} for each point~$x\in \CVR_k(H, S)$ (resp. $\CVR_k(H', S)$)
    near $E$.
    Likewise for adjacent regions of $\mCVD_k(S)$.
\end{restatable}
\begin{proof}
    Let $p\in E$ and let $a_1 \colpreceq_p \cdots \colpreceq_p
    a_k\colpreceq_p a_{k+1}$ be the ordering of
    the first $k+1$ colors with respect to the minimal dominance
    relation~$\colpreceq_p$.
    Since $p$ is on the boundary of both $\CVR_k(H, S)$ and $\CVR_k(H', S)$,
    it follows that $a_k \colequiv_p a_{k+1}$,
    and $p\in \Bisector{a_k}{a_{k+1}}$.
    
    Let $x \in \CVR_k(H, S)\cap N(p)$ and $x' \in \CVR_k(H', S)\cap N(p)$ be points in a very small
    neighborhood $N(p)$ around $p$.
    By continuity, colors $a_1 ...a_{k-1}$ are the first $k-1$ colors with
    respect to both the minimal dominance
    relations~$\colpreceq_x$ and~$\colpreceq_{x'}$, while the $k$-th color
    is $a_k$ for one and $a_{k+1}$ for the other.
    The claim follows, where $a=a_k$ and
    $a'=a_{k+1}$ if $\Dominance{a}{a'} $ lies in the same side
    of $\Bisector{a}{a'}$ as $x$.
\end{proof}

We conclude that the order-$k$ minimal diagram~$\CVD_k(S)$ is a plane graph that
partitions the plane 
into order-$k$ minimal regions~$\CVR_k(H,S)$,
for  $H \subseteq K$ with $|H| = k$,
while the order-$k$ maximal diagram~$\mCVD_k(S)$ is a plane graph that partitions $\Plane$ 
into order-$k$ maximal regions~$\mCVR_k(H,S)$.
The edges of $\CVD_k(S)$ (resp. $\mCVD_k(S)$), called \emph{Voronoi edges},
are portions of minimal (resp. maximal) color bisectors that  bound their incident 
order-$k$ color Voronoi regions. 
The Voronoi vertices  are  intersection points of related color
bisectors. See Figure~\ref{fig:cvd} for a concrete example of line-segment sites. 

As with 
ordinary order-$k$ diagrams, the vertices of $\CVD_k(S)$ and
$\mCVD_k(S)$ are characterized as
\emph{new} 
and \emph{old}. 
A vertex of $\CVD_k(S)$ (resp. $\mCVD_k(S)$) is called \emph{new}
if it does not appear in $\CVD_{k-1}(S)$ (resp. $\mCVD_{k-1}(S)$);
and \emph{old}, otherwise. 
See Figures~\ref{fig:cvd}(a--c) and~\ref{fig:mcvd}(a--c),
where new vertices are marked by small squares.

\begin{restatable}{lemma}{vertexcharacterization} \label{lem:vertex_characterization}
  A new vertex $v$ of $\CVD_{k-1}(S)$ is an old vertex of~$\CVD_k(S)$,
  and $v$
  is not a vertex of $\CVD_{k+1}(S)$. No Voronoi edge of $\CVD_{k}(S)$ remains in
  $\CVD_{k+1}(S)$.
  Analogous statements hold for $\mCVD_{k}$.
\end{restatable}
\begin{proof}
    Let $x$ be a Voronoi vertex of~$\CVD_k(S)$, where  
    $x \in \cl\CVR_k(H_1, S) \cap \cl\CVR_k(H_2,S) \cap \cl\CVR_k(H_3,S)$.
    By Lemma~\ref{lem:neighboring-regions} 
    we have $|H_i \cap H_j| = k-1$ and $|H_i \setminus H_j| = 1$ 
    for any pair of distinct indices~$i,j \in \{1,2,3\}$.
    This implies only two possibilities:
    either $|H_1 \cap H_2 \cap H_3| = k-1$ 
    or $|H_1 \cap H_2 \cap H_3| = k-2$. 
    In the former case, 
    we have
    $H_1 \cap H_2 \cap H_3 = \{a_1, \ldots, a_{k-1}\}$,
    $H_1 \cup H_2 \cup H_3 = \{a_1, \ldots, a_{k+2}\}$, and
    $ a_{k-1} \colprec_x a_{k} \colequiv_x a_{k+1} \colequiv_x a_{k+2}
    \colprec_x a_{k+3}$,
    thus, $x$ must be \emph{new} in~$\CVD_k(S)$ as $x\not\in \CVD_{k-1}(S)$.
    In the latter case, 
    we have
    $H_1 \cap H_2 \cap H_3 = \{a_1, \ldots, a_{k-2}\}$,
    $H_1 \cup H_2 \cup H_3 = \{a_1, \ldots, a_{k+1}\}$, and
    $a_{k-2} \colprec_x a_{k-1} \colequiv_x a_{k} \colequiv_x a_{k+1} \colprec_x
    a_{k+2}$.
    Then, $x$ is a vertex of $\CVD_{k-1}(S)$, and consequently \emph{old} in~$\CVD_k(S)$.
    It is now obvious that an 
    old vertex of~$\CVD_k(S)$ is a new vertex of~$\CVD_{k-1}(S)$, and 
    is not a vertex of $\CVD_{j}(S)$ for any $j \notin \{k,k-1\}$.
    Similarly, no Voronoi edge of $\CVD_{k}(S)$ can remain in
    $\CVD_{k+1}(S)$.
\end{proof}


\subparagraph*{Farthest color and Hausdorff Voronoi diagrams.}
The case of $k=m-1$ is of a special interest.
The order-$(m{-}1)$ minimal diagram~$\CVD_{m-1}(S)$ is determined 
by $m-1$ nearest colors, or equivalently the farthest color, with respect to the minimal dominance~$\colpreceq_x$,
resulting in the \emph{farthest color Voronoi diagram}~$\FCVD(S)$;
the  $\mCVD_{m-1}(S)$ is determined by $m-1$ farthest colors, or  equivalently the nearest color,
with respect to the maximal dominance~$\mcolpreceq_x$,
resulting in the  \emph{Hausdorff Voronoi diagram}~$\HVD(S)$.
For a color $a\in K$, its \emph{farthest color Voronoi region} is
\(
  \FCVR(a, S) := \bigcap_{b \in K \setminus \{a\}}
  \Dominance{b}{a} = \CVR_{m-1}(K \setminus \{a\}, S);
\)
while its \emph{Hausdorff Voronoi region} is
\(
 \HVR(a,S) := \bigcap_{{\substack{b \in K \setminus \{a\}}}} \mDominance{b}{a} = \mCVR_{m-1}(K \setminus \{a\},S).
 \)

The structure of these Voronoi regions becomes evident once we overlay them
with $\VD(S_a)$, and $\FVD(S_a)$, respectively;
see the red regions in Figures~\ref{fig:fcvd_vs_fcvd_star}(a--b).
Each face~$f$ of~$\FCVR(a, S)\cap \VD(S_a)$ is associated with  a site~$p \in
S_a$, where color $a$ is the farthest, with respect to the
minimal dominance relation~$\colpreceq_x$, for any point $x \in f$, and site $p$ is the nearest
at~$x$ among the sites in~$S_a$. 
The resulting refined diagram is denoted by~$\FCVD^*(S)$; see Figure~\ref{fig:cvd}(g).
Analogously, 
each face~$\bar{f}$ of $\HVR(a,S)\cap \FVD(S_a)$ is associated with a
site~$p\in S$ such that the color~$a$ of $p$ is the nearest, with respect to the maximal dominance
relation~$\mcolpreceq_x$, for any point~$x\in \bar{f}$, while 
site $p \in S_a$ is the farthest at~$x$ among the sites in~$S_a$.
The resulting refined diagram is denoted by~$\HVD^*(S)$; see
Figure~\ref{fig:mcvd}(g).

\begin{figure}[ht]
    \centering
    \includegraphics[height=11.4cm]{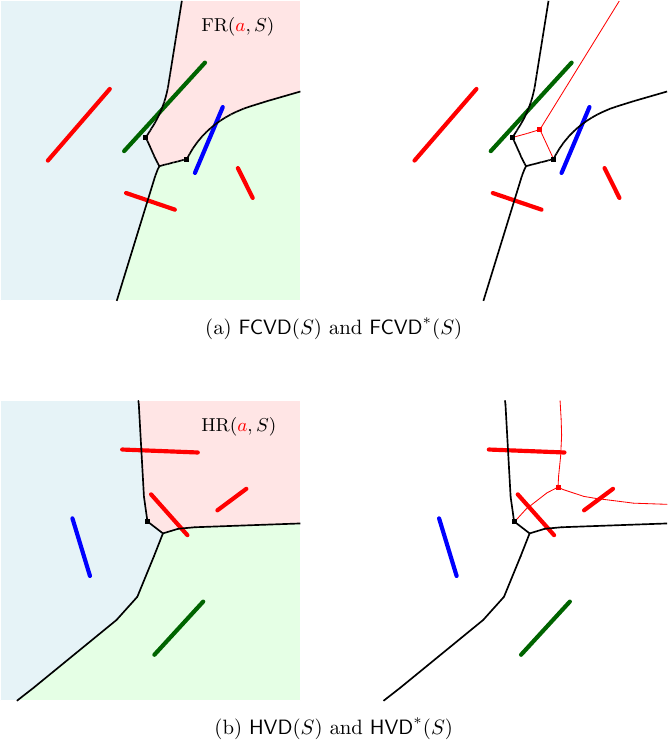}
    \caption{The farthest color and Hausdorff Voronoi diagrams next to their refined counterparts.
    In (a) the red edges 
    are $\VD(S_{a}) \cap \FCVR(a,S)$. In (b) the red edges 
    are $\FVD(S_{a}) \cap \HVR(a,S)$.
    }
    \label{fig:fcvd_vs_fcvd_star}
\end{figure}

Our abstract definition of~$\FCVD(S)$ and~$\HVD(S)$ 
encompasses all concrete cases that have been studied
before~\cite{mpsw-fcvd:ca-25,p-hvdpcp-04,b-tbiafcvdls-14,b-lsfcvd-12,bop-hocvdccsf-25},
including the farthest polygon Voronoi diagram~\cite{cegghlln-fpvd-11},
and any other case of generalized sites or metrics under the AVD model.

 


\subparagraph*{Refined order-$k$ color diagrams.}
%
%
We can analogously refine the order-$k$ diagrams~$\CVD_k(S)$ and $\mCVD_k(S)$,
denoted by~$\CVD^*_k(S)$ and $\mCVD^*_k(S)$, respectively, for each~$1\leq k\leq m$,
where $\CVD_m(S)$ and $\mCVD_m(S)$ consist of a single region
$\CVR_m(K,S)=\mCVR_m(K,S)=\Plane$.
%
%

 For each nonempty region~$R = \CVR_k(H, S)$ of~$\CVD_k(S)$,
 we overlay~$\FCVD^*(S_H)$ inside~$R$, where $S_H = \bigcup_{a\in H} S_a$.
 The resulting refined diagram is denoted~$\CVD^*_k(S)$.
 Recall that $\FCVD(S_H)=\CVD_{k-1}(S_H)$.
 Note that each face~$f$ of~$\CVD^*_k(S)$ is associated with
 a set~$H \subseteq K$ of $k$~colors and a site~$p\in S_a$ for some
 color~$a\in H$
 such that, at any point~$x\in f$,
 $H$ is the common set of $k$ nearest colors,
 color~$a$ is the $k$-th nearest color with respect to the minimal dominance~$\colpreceq_x$, and
 $p \in S_a$ is the nearest site of color $a$; 
 that is, $ f \subseteq \CVR_k(H,S) \cap \FCVD(S_H) \cap \VR(p, S_a)$.
 See Figures~\ref{fig:cvd}(d--g), where faces of $\CVD^*_{k}(S)$ associated
 to the sites $p$ and $q$
 are shaded purple and red, respectively.

 For each nonempty region~$\overline{R} = \mCVR_k(H, S)$ of~$\mCVD_k(S)$,
 we overlay~$\HVD^*(S_H)$ inside~$\overline{R}$.
 The resulting refined diagram is denoted by~$\mCVD^*_k(S)$.
 Similarly as above, each face~$\bar{f}$ of~$\mCVD^*_k(S)$ is associated with 
 a subset~$H \subseteq K$ of $k$~colors and a site~$p\in S_a$ for some~$a\in H$
 such that, at any point~$x\in \bar{f}$,
 $H$ is the common set of $k$ farthest colors,
 color~$a$ is the $k$-th farthest color with respect to the maximal dominance~$\mcolpreceq_x$, and
 $p \in S_a$ is the farthest site among those in~$S_a$;
 that is, 
 $\bar{f} \subseteq \mCVR_k(H,S) \cap \HVD(S_H) \cap \FVR(p, S_a)$.
 See Figures~\ref{fig:mcvd}\mbox{(d--g)}, where faces of $\mCVD^*_{k}(S)$ associated
 to the sites $p$ and $q$
 are also shaded purple and red.


Note $\CVD^*_m(S) = \FCVD^*(S)$ and $\mCVD^*_m(S) = \HVD^*(S)$, see
Figures~\ref{fig:cvd}(g) and~\ref{fig:mcvd}(g), 
whereas $\CVD_{m-1}(S) = \FCVD(S)$ and $\mCVD_{m-1}(S) = \HVD(S)$, see Figures~\ref{fig:cvd}(c) and~\ref{fig:mcvd}(c).

\begin{figure}[p]
    \centering
    \includegraphics[width=0.98\textwidth]{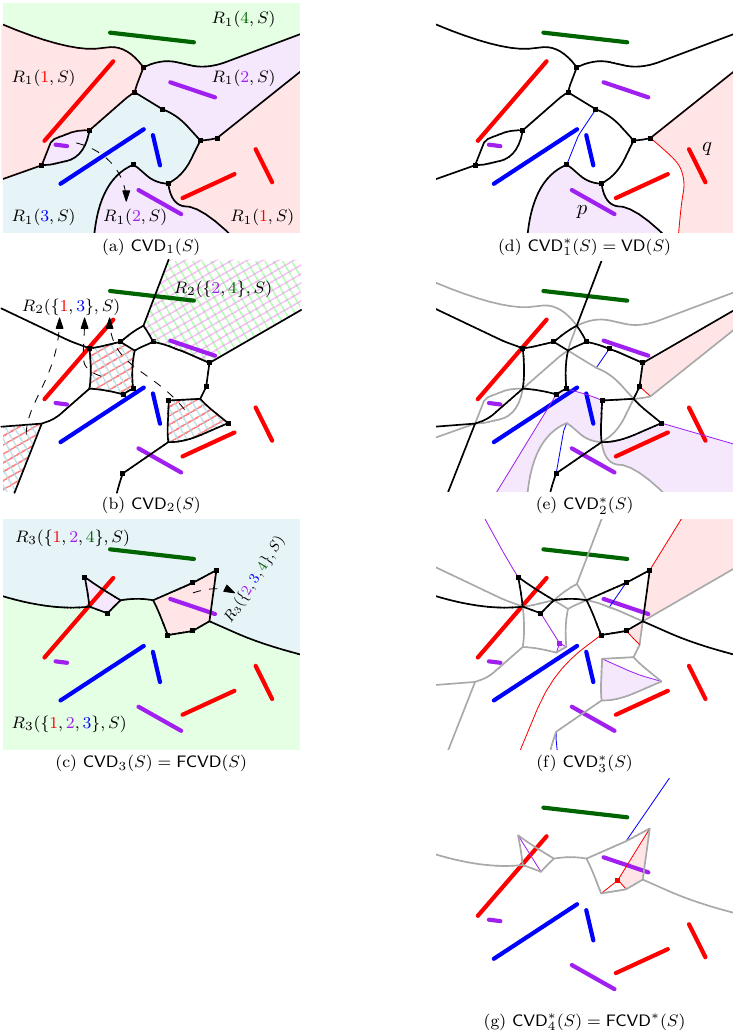}
    \caption{The minimal color Voronoi diagrams 
        of all orders. 
        The colors $K = \{{\color{red}1},{\color{Purple}2},{\color{blue}3},{\color{DarkGreen}4}\}$ are red, purple, blue and green. 
        New $2$-chromatic edges are in black, old $2$-chromatic edges in gray,
        $1$-chromatic edges in their own color, and new vertices are small squares.
    }
    \label{fig:cvd}
\end{figure}

\begin{figure}[p]
    \centering
    \includegraphics[width=0.98\textwidth]{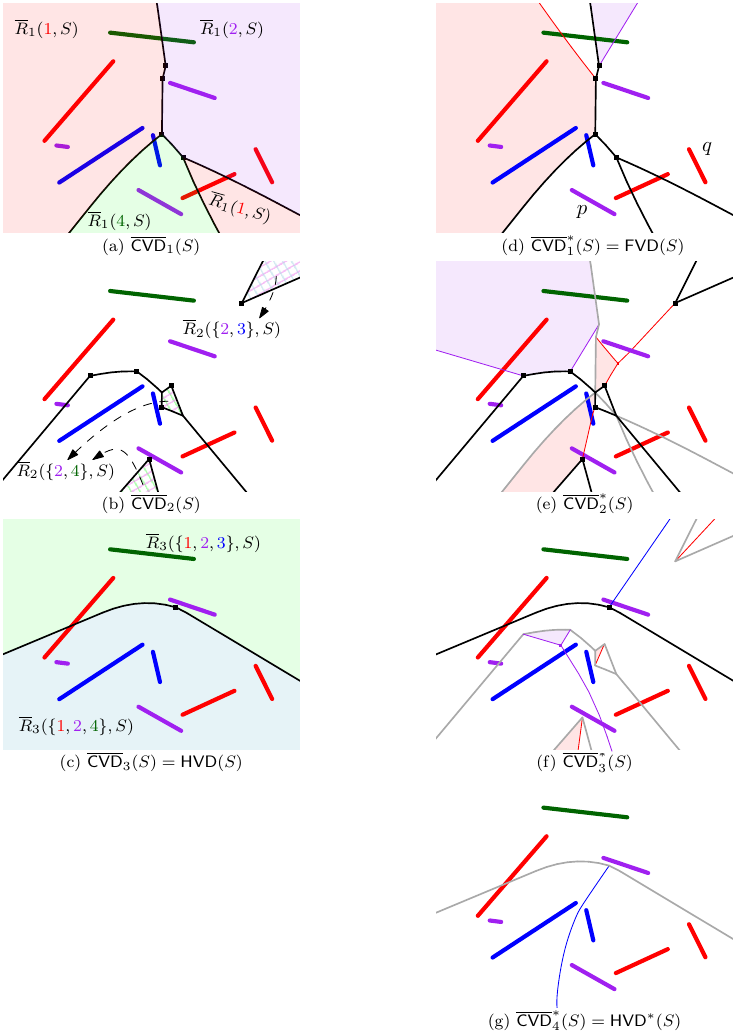}
    \caption{The maximal color Voronoi diagrams 
        of all orders, 
        with the same sites and style for $c$-chromatic features as Figure~\ref{fig:cvd}.
        Faces in $\mCVD^*_{k}(S)$
        associated to $p$ and $q$
        are shaded purple and red. 
    }
    \label{fig:mcvd}
\end{figure}

\subparagraph*{Distinguishing chromatic features.}
We can distinguish different features of the color diagrams by  the number of their
defining colors.
Any edge or vertex of a color diagram
is called \emph{\mbox{$c$-chromatic}}, where $c\in\{1,2,3\}$ is the number of
distinct colors of the sites defining it.
In $\CVD_k(S)$ and $\mCVD_k(S)$,  Voronoi edges and their  degree-$2$ vertices
are bichromatic ($2$-chromatic), while the Voronoi vertices are all trichromatic ($3$-chromatic).
No monochromatic ($1$-chromatic) features exist in $\CVD_k(S)$ or $\mCVD_k(S)$.

As with ordinary order-$k$ diagrams, the following properties are derived from definitions. 
\begin{restatable}{lemma}{overlayCVDkstar}
  \label{lem:overlay_CVDk_star}
  For each face $R\subseteq \CVR_k(H,S)$ of~$\CVD_{k}(S)$, 
  it holds: $\FCVD(S_H) \cap R = \CVD_{k-1}(S) \cap R$.
  For each face $\overline{R}\subseteq \mCVR_k(H,S)$ of~$\mCVD_k(S)$: 
  $\HVD(S_H) \cap \overline{R} = \mCVD_{k-1}(S) \cap \overline{R}$.
  Thus,
  superimposing $\CVD_k(S)$ and $\CVD_{k-1}(S)$ results in $\CVD^*_{k}(S)$
  with its monochromatic vertices and edges removed.
  Analogous statements hold for the maximal diagrams. 
\end{restatable}
\begin{proof}
    Since $R\subseteq \CVR_k(H, S)$ and $|H|=k$, it follows that $\CVD_{k-1}(S) \cap R =
    \CVD_{k-1}(S_H) \cap R$. But $\CVD_{k-1}(S_H)= \FCVD(S_H)$, thus, the first
    claim follows.
    Analogously, for 
    $\overline{R} \subseteq \mCVR_k(H, S)$, as $\HVD(S_H) = \mCVD_{k-1}(S_H)$,
    while $\mCVD_{k-1}(S)\cap \overline{R}=\mCVD_{k-1}(S_H)\cap \overline{R}$.
    %
    
    The superimposition property  follows since $\CVD^*_k(S) \cap
    R=\FCVD^*(S_H)\cap R$, and $\FCVD(S_H)\cap R= \CVD_{k-1}(S_H)$ as already shown.
\end{proof}


\evanthia{If we do not use new and old in the following then the def  can move to
    the appendix or where we need it first time.} 
A vertex or an edge of $\CVD^*_k(S)$ or of $\mCVD^*_k(S)$ is called \emph{new} 
if it does not appear in $\CVD^*_{k-1}(S)$ or in $\mCVD^*_{k-1}(S)$, respectively;
and \emph{old}, otherwise. 
Note that vertices and edges of $\CVD^*_1(S)$ and $\mCVD^*_1(S)$ are all new,
while those of $\CVD^*_m(S)$ and $\mCVD^*_m(S)$ are all old.


\begin{restatable}{lemma}{consecutiveorders} \label{lem:consecutive_orders}
 A $c$-chromatic vertex or  edge  
 appears 
in  the diagrams $\CVD^*_k(S)$ or~$\mCVD^*_k(S)$ for $c$~consecutive
 orders, and in ~$\CVD_k(S)$ or~$\mCVD_k(S)$ for $c{-}1$~consecutive orders.
\end{restatable} 
\begin{proof}
    Consider a $3$-chromatic new vertex $v$ of $\CVD_k(S)$.
    By Lemma~\ref{lem:vertex_characterization}, 
    $v$ is an old vertex in $\CVD_{k+1}(S)$, and it does not appear in
    $\CVD_{k+2}(S)$. In fact, it cannot appear in $\CVD_{j}(S)$ for any $j\notin
    \{k,k+1\}$ (see proof of Lemma~\ref{lem:vertex_characterization}).
    However, by
    Lemma~\ref{lem:overlay_CVDk_star}, $v$ appears 
    in $\CVD_{k+2}^*(S)$. Thus, the claim follows for $c=3$.
    \deleted{
        any other order. 
        Consider the  refined diagrams $\CVD^*_k(S)$, $\CVD^*_{k+1}(S)$, and $\CVD^*_{k+2}(S)$.
        By Lemma~\ref{lem:overlay_CVDk_star}, these correspond to the superimpositions
        of the following diagrams: $\{\CVD_k(S), \CVD_{k-1}(S)\}$, $\{\CVD_{k+1}(S), \CVD_k(S)\}$,
        and $\{\CVD_{k+2}(S), \CVD_{k+1}(S)\}$, respectively.
        Since $v$ is a 
        vertex in $\CVD_k(S)$, it appears in $\CVD^*_k(S)$ and $\CVD^*_{k+1}(S)$.
        Since $v$ is a 
        vertex in $\CVD_{k+1}(S)$, it also appears in $\CVD^*_{k+2}(S)$.
    }
    
    \evanthia{In the main text we could also say: Analogously for $c=2,1$ and
        skip.}
    
    Next, consider a $2$-chromatic edge $e$ in $\CVD_k(S)$ (and any $2$-chromatic vertex $v$ that $e$
    may contain).
    By Lemma~\ref{lem:vertex_characterization},
    the edge $e$ 
    does not appear in $\CVD_{k+1}(S)$, nor can it
    appear in any $\CVD_{j}(S)$ for $j\neq k$; 
    however, it appears in
    $\CVD^*_{k+1}(S)$, by Lemma~\ref{lem:overlay_CVDk_star}.
    Thus, the claim follows for $c=2$.
    
    Finally, consider a $1$-chromatic edge $e$ (resp. vertex)  defined by sites of color $a$.
    By the definition of the refined diagrams, $e$ appears in $\CVD^*_k(S)$
    if and only if $a$ is the unique $k$-th nearest color of the points in $e$.
    Thus, it appears in exactly $1$ refined diagram.
\end{proof}


A full characterization of the $c$-chromatic features
of~$\CVD^*_k(S)$ and $\mCVD^*_k(S)$ can be established
based on the ordering of sites in~$S$ with respect to~$\spreceq_x$
and the ordering of colors in~$K$ with respect to~$\colpreceq_x$.

\begin{restatable}{lemma}{edgetype} \label{lem:edge_types}
    For any~$x\in\Plane$, let $p_1 \spreceq_x \cdots \spreceq_x p_n$ and
    $a_1 \colpreceq_x \cdots \colpreceq_x a_m$ be 
    the ordering of sites in~$S$ with respect to the relations~$\spreceq_x$
    and the ordering of colors in~$K$ with respect to~$\colpreceq_x$, respectively.
    \begin{itemize} 
        \item 
        Point~$x$ lies 
        on a bichromatic edge of~$\CVD^*_k(S)$
        if and only if
        $a_{l-1} \colprec_x a_l \colequiv_x a_{l+1} \colprec_x a_{l+2}$
        for $l\in \{k-1, k\}$.
        \item
        Point~$x$ lies 
        on a monochromatic edge of~$\CVD^*_k(S)$
        if and only if
        $a_{k-1} \colprec_x a_k \colprec_x a_{k+1}$ and
        $p_j \sequiv_x p_{j+1}$ for some~$j$
        such that $p_j, p_{j+1} \in S_{a_k}$ and $p_i \notin S_{a_k}$ for all~$i<j$.
    \end{itemize}
\end{restatable}
\begin{proof}
    First, consider the bichromatic edges of $\CVD^*_k(S)$.
    By Lemma~\ref{lem:overlay_CVDk_star}, these are exactly the edges
    of $\CVD_k(S)$ and of $\CVD_{k-1}(S)$.
    It follows from Lemma~\ref{lem:basics_edge} and the general position assumption, that a point $x$ on an edge $e$ of $\CVD_{k}(S)$ (resp. $\CVD_{k-1}(S)$) lies on the common boundary of two regions $\CVR_k(H, S)$ and $\CVR_k(H', S)$ (resp. $\CVR_{k-1}(H, S)$ and $\CVR_{k-1}(H', S)$).
    By Lemma~\ref{lem:neighboring-regions},
    such an edge $e$ is a portion of the bisector $\Bisector{a}{a'}$
    of the 
    colors $a \in H \setminus H'$ and $a' \in H' \setminus H$.
    For a point $x$ near $e$ contained in $\CVR_k(H, S)$ (resp. $\CVR_{k-1}(H, S)$),
    $a$ and $a'$ are the $k$-th and $(k{+}1)$-th nearest colors (resp. $(k{-}1)$-th and $k$-th)
    with respect to $\colpreceq_x$.
    Thus, a point $x$ lying on an edge of $\CVD_{k}(S)$ (resp. $\CVD_{k-1}(S)$)
    implies $a \colequiv_x a'$,
    where $a = a_k$ and $a' = a_{k+1}$ (resp. $a = a_{k-1}$ and $a' = a_{k}$).
    Conversely, if these equivalence conditions hold at $x$, any sufficiently small neighborhood of $x$ intersects both adjacent regions, thus $x$ must lie on their common boundary.
    
    Next, consider a monochromatic edge $e$ of $\CVD^*_k(S)$,
    which is a monochromatic edge of $\FCVD^*(S_H) \cap \CVR_k(H, S)$, 
    $H \subseteq K$.
    Furthermore, $e$ lies strictly inside the region~$\FCVR(a, S_H)$ for some $a \in H$.
    A point $x$ is contained in $\CVR_k(H, S) \cap \FCVR(a, S_H)$
    if $H$ is the set of $k$ nearest colors from $x$
    and $a$ is the unique farthest color in $H$
    with respect to $\colpreceq_x$,
    implying that $a = a_k$ and $a_{k-1} \colprec_x a_k \colprec_x a_{k+1}$.
    A 1-chromatic edge $e$ of $\FCVD^*(S_H)$ inside $\FCVR(a, S_H)$
    is a portion of an edge of $\VD(S_a)$.
    Under the general position assumption, $x$ lying on an edge of $\VD(S_a)$ means that
    $x$ lies on the bisector $\Bisector{p_j}{p_{j+1}}$
    of the nearest sites $p_j, p_{j+1} \in S_a$ with respect to $\spreceq_x$.
    Therefore, $p_j \sequiv_x p_{j+1}$ such that no $p_i$ for $i<j$ belongs to $S_{a_k}$. 
    The converse holds analogously: the strict color order ensures that the neighborhood around $x$ is contained in $\CVR_k(H, S)$, while the site equivalence implies it must lie on the bisector $\Bisector{p_j}{p_{j+1}}$.
\end{proof}

\begin{restatable}{lemma}{vertextype} \label{lem:vertex_types}
    For any~$x\in\Plane$, let $p_1 \spreceq_x \cdots \spreceq_x p_n$ and
    $a_1 \colpreceq_x \cdots \colpreceq_x a_m$ be
    the ordering of sites in~$S$ with respect to the relations~$\spreceq_x$
    and the ordering of colors in~$K$ with respect to~$\colpreceq_x$, respectively.
    \begin{itemize} 
        \item Point $x$ is a $3$-chromatic vertex of $\CVD^*_k(S)$ if and only if
        $a_l \colequiv_x a_{l+1} \colequiv_x a_{l+2}$
        for some $l \in\{{k-2},k-1, k\}$.
        \item Point $x$ is a $2$-chromatic vertex of $\CVD^*_k(S)$ if and only if
        $a_l \colequiv_x a_{l+1}$
        for some~$l \in \{k-1, k\}$, and $p_j \sequiv_x p_{j+1} \sequiv_x p_{j+2}$
        such that $p_j, p_{j+1}, p_{j+2} \in S_{a_l} \cup S_{a_{l+1}}$ and
        $p_i \notin S_{a_l} \cup S_{a_{l+1}}$ for all~$i<j$.
        \item Point $x$ is a $1$-chromatic vertex of $\CVD^*_k(S)$ if and only if
        $a_{k-1} \colprec_x a_k \colprec_x a_{k+1}$ and \linebreak[4]
        $p_j \sequiv_x p_{j+1} \sequiv_x p_{j+2}$ such that
        $p_j, p_{j+1}, p_{j+2}\in S_{a_k}$ and $p_i \notin S_{a_k}$ for all~$i<j$.
    \end{itemize}
\end{restatable}
\begin{proof}
    First, consider the $3$-chromatic vertices of $\CVD^*_k(S)$.
    By Lemma~\ref{lem:overlay_CVDk_star}, these are exactly the vertices
    of $\CVD_k(S)$ and $\CVD_{k-1}(S)$.
    Recall from the proof of Lemma~\ref{lem:vertex_characterization} that
    $x$ is a new vertex of $\CVD_k(S)$ (resp. $\CVD_{k-1}(S)$)
    if $a_k \colequiv_x a_{k+1} \colequiv_x a_{k+2}$
    (resp. $a_{k-1} \colequiv_x a_k \colequiv_x a_{k+1}$),
    and $x$ is an old vertex if $a_{k-1} \colequiv_x a_k \colequiv_x a_{k+1}$
    (resp. $a_{k-2} \colequiv_x a_{k-1} \colequiv_x a_k$).
    Thus, if $x$ is a $3$-chromatic vertex then
    $a_l \colequiv_x a_{l+1} \colequiv_x a_{l+2}$
    for some $l \in \{k-2, k-1, k\}$. The reverse implication can be proved analogously, as these equivalences imply the neighborhood of $x$ intersects the three regions, defining a vertex.
    
    Next, consider the $2$-chromatic vertices of $\CVD^*_k(S)$,
    which must lie on the $2$-chromatic edges of $\CVD^*_k(S)$.
    By Lemma~\ref{lem:edge_types}, if a point $x$ is on a $2$-chromatic edge of $\CVD^*_k(S)$,
    then $a_l \colequiv_x a_{l+1}$
    for some $l \in \{k-1, k\}$, and thus $x \in \Bisector{a_l}{a_{l+1}}$.
    Since $x$ is a vertex of $\Bisector{a_l}{a_{l+1}}$,
    it is a vertex of $\VD(S_{a_l} \cup S_{a_{l+1}})$
    whose three nearest sites in $S_{a_l} \cup S_{a_{l+1}}$, with respect to $\spreceq_x$,
    include at least one site of color $a_l$ and one of color $a_{l+1}$.
    Under the general position assumption,
    a point $x$ being a vertex of $\VD(S_{a_l} \cup S_{a_{l+1}})$
    implies it is equidistant to three sites $p_j \sequiv_x p_{j+1} \sequiv_x p_{j+2}$
    that are the nearest in $S_{a_l} \cup S_{a_{l+1}}$ with respect to $\spreceq_x$.
    Hence, $p_j \sequiv_x p_{j+1} \sequiv_x p_{j+2}$
    such that no $p_i$ for $i<j$ belongs to $S_{a_l} \cup S_{a_{l+1}}$. Conversely, satisfying these equivalences implies that $x$ lies on the common intersection of the three related bisecting curves.
    
    Finally, consider the $1$-chromatic vertices of $\CVD^*_k(S)$.
    As with the $1$-chromatic edges from Lemma~\ref{lem:edge_types},
    a $1$-chromatic vertex $x$ has a unique $k$-th nearest color $a_k$,
    so $a_{k-1} \colprec_x a_k \colprec_x a_{k+1}$,
    and $x$ is a vertex of $\VD(S_{a_k})$.
    Under the general position assumption, being a vertex of $\VD(S_{a_k})$ means
    $x$ is equidistant from three sites $p_j \sequiv_x p_{j+1} \sequiv_x p_{j+2}$
    that are the nearest in $S_{a_k}$ with respect to $\spreceq_x$.
    Thus, no $p_i$ for $i<j$ belongs to $S_{a_k}$. For the converse, the strict color ordering guarantees the neighborhood of $x$ lies strictly inside the region where $a_k$ is the $k$-th nearest color, while the site equivalence implies it is a vertex of the nearest-site Voronoi diagram.
\end{proof}

Using the maximal dominance relations~$\mcolpreceq_x$, $\mcolprec_x$, and
$\mcolequiv_x$, the vertices and edges of the maximal diagram~$\mCVD^*_k(S)$ can be 
characterized analogously.



\subparagraph*{A concrete example: the order-$k$ polygon Voronoi diagram.}
\label{subsec:polygons}
Let $\mathcal{P}$ be a set of $m$ disjoint simple polygons with a total of $n$
vertices. Each polygon induces a distinct color in set $K$.
We can cast the order-$k$ polygon Voronoi diagram as a concrete case of an  order-$k$
abstract color (minimal) Voronoi diagram, $\CVD_k(S)$,  by selecting the set  $S$ of sites so that
the underlying  bisector system is admissible.
There are two ways  to do so, each resulting in a slightly
different (but equally valid) minimal refined diagram:
(1) let $S$ be  a set of points and open line segments derived from the
vertices and edges of the polygons in $\mathcal{P}$; and
(2) let  $S$ be a set of line segments derived from the polygonal edges  in
$\mathcal{P}$.
In the latter case, if two segments $p$ and $q$ have a
common endpoint, then replace  the two-dimensional portion of their bisector
by the piece of their angular bisector contained in their 
equidistant area~\cite{adk-flsvd-06}, see
Figure~\ref{fig:segment_bisector}. 
In both cases, the resulting families of bisecting curves   are 
admissible as they satisfy the required axioms.
Both definitions yield an identical $\CVD_k(S)$, however,
the corresponding refined diagrams $\CVD^*_k(S)$ can have different monochromatic features.

The maximal diagram $\mCVD_k(S)$ depends on which convention underlies the choice of $S$.
Under convention~(1) 
$\mCVD_k(S)$ 
reduces to a point maximal diagram $\mCVD_k(S')$, where $S'\subset S$ consists of
the vertices in $\mathcal{P}$; this is not the case under convention~(2).

\begin{figure}[hb]
    \centering
    \includegraphics[height=2.75cm]{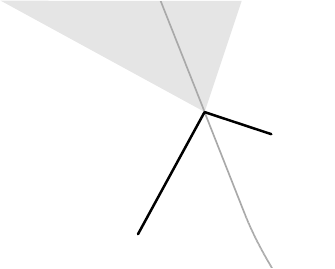}
    \caption{%
        The angular bisector of two segments over their equidistant area shown in
        gray. 
    }
    \label{fig:segment_bisector}
\end{figure}

\section{Combinatorial complexity}
\label{sec:complexity}

In this section, we
establish the upper bound $4k(n-k)-2n$, which is tight, for both
minimal and maximal order-$k$ abstract color Voronoi diagrams. 
Our proof is based
on two techniques: 
(1)~a random sampling technique, 
known as the \emph{Clarkson--Shor technique}~\cite{cs-arscgII-89} with a 
\emph{colorful} extension in~\cite{bop-hocvdccsf-25}; and (2)~a 
combinatorial analysis on the circular sequences of permutations induced by the unbounded
edges of the order-$k$ color Voronoi diagrams. 
%
%
Part~(1) adapts the technique from \cite{bop-hocvdccsf-25} to abstract Voronoi
diagrams and results in Lemma~\ref{lem:complexity}.
Ppart~(2) is detailed in Section~\ref{sec:SumU} and its results are 
summarized in Lemma~\ref{lem:sumU}.

We first introduce some notation.
Consider the \emph{new}  vertices and the \emph{new} unbounded edges of 
$\CVD^*_k(S)$ and~$\mCVD^*_k(S)$, $1\leq k\leq m$.
Recall that a  feature of an  order-$k$ diagram is called \emph{new} if it  does not
appear in the diagram of 
order-$(k{-}1)$; and \emph{old}, otherwise.
Their numbers are denoted as follows.

\begin{itemize}
\item Let $\bm{v_{c,k}} = v_{c,k}(S)$ and $\bar{v}_{c,k} = \bar{v}_{c,k}(S)$ denote  the number of
  new $c$-chromatic vertices of~$\CVD^*_k(S)$ and of~$\mCVD^*_k(S)$, respectively.
\item Let $\bm{u_{c,k}} = u_{c,k}(S)$ and $\bar{u}_{c,k} = \bar{u}_{c,k}(S)$ denote the number of
  new $c$-chromatic unbounded edges of~$\CVD^*_k(S)$ and of~$\mCVD^*_k(S)$, respectively.
  Note that an edge that is unbounded in both directions 
  is counted as two distinct unbounded edges.
  Also  $u_{c,m} = \bar{u}_{c,m} = 0$.
\end{itemize}


We now consider the first part of our proof  whose results are summarized
in Lemma~\ref{lem:complexity}.
The
Clarkson--Shor technique deals in general with a set system based on
three ingredients: a set of $n$ objects $O$, a set of configurations
$\conf(O)$,  each  defined by exactly $d$~objects, and
a conflict relation  $\chi \subseteq O \times \conf(O)$,
such that none of the $d$~objects defining~$f$ are in conflict with~$f$.
The set system defined by the triple~$(O, \conf(O), \chi)$ is called a
\emph{CS-structure}.

Given a CS-structure,
suppose that the objects in~$O$ are colored
by a color assignment~$\kappa \colon O \to K = \{1, \ldots, m\}$,
where $K$ is the set of $m\leq n$~colors.
In~\cite{bop-hocvdccsf-25}
a general framework was introduced, called the \emph{colorful Clarkson--Shor framework}, which
creates color-augmented
CS-structures $(K,\conf(O,\kappa), \chi_\kappa)$, such that 
the objects are  colors,
$\conf(O,\kappa)
\subseteq \conf(O)$ are \emph{colored configurations}, and  $\chi_\kappa$ is a
color-to-configuration conflict relation.
The color-augmented CS-structure remains a CS-structure, therefore, the general random sampling
technique of  Clarkson and Shor~\cite{cs-arscgII-89} can still  be applied.

\deleted{
    Bae, Oliver, and Papadopoulou~\cite{bop-hocvdccsf-25} introduced
    a general framework, called the \emph{colorful Clarkson--Shor framework}, 
    that, based on  $(O,\conf(O), \chi)$ and $\kappa$, derives
    \emph{colored configurations}~$\conf(O,\kappa) \subseteq \conf(O)$,
    and a color-to-configuration conflict relation~$\chi_\kappa$, 
    deriving  a color-augmented CS-structure
    in which colors in~$K$ are considered as \emph{objects}.
    The general random sampling framework by Clarkson and Shor~\cite{cs-arscgII-89}
    can then be applied to those colored configurations.
}

\deleted{
    Recall the following definition from Section~\ref{sec:complexity}.
    \begin{itemize}
        \item Let $v_{c,k} = v_{c,k}(S)$ and $\bar{v}_{c,k} = \bar{v}_{c,k}(S)$ denote  the number of
        new $c$-chromatic vertices of~$\CVD^*_k(S)$ and of~$\mCVD^*_k(S)$, respectively.
        \item Let $u_{c,k} = u_{c,k}(S)$ and $\bar{u}_{c,k} = \bar{u}_{c,k}(S)$ denote the number of
        new $c$-chromatic unbounded edges of~$\CVD^*_k(S)$ and of~$\mCVD^*_k(S)$, respectively.
        Note that an edge that is unbounded in both directions 
        is counted as two distinct unbounded edges.
        Note also that $u_{c,m} = \bar{u}_{c,m} = 0$.
    \end{itemize}
}

We first observe the following from well-known properties
of nearest and farthest abstract Voronoi diagrams~\cite{k-cavd-89,mmr-fsavd-01},
under the general position assumption.
\begin{restatable}{lemma}{vOuO} \label{lem:v0_u0}
    For any subset~$S'\subseteq S$, the following hold:
    \[ \sum_{1\leq c\leq 3} v_{c,1}(S') = 2|S'| - 2 - \sum_{1\leq c\leq 2} u_{c,1}(S')
    \quad\text{and}\quad
    \sum_{1\leq c\leq 3} \bar{v}_{c,1}(S') = \sum_{1\leq c\leq 2} \bar{u}_{c,1}(S') - 2.
    \]
\end{restatable}
\begin{proof}
    The quantities $\sum_{c} v_{c,1}(S')$ and $\sum_c u_{c,1}(S')$ count the vertices 
    and the unbounded edges, respectively, 
    of the nearest Voronoi diagram~$\CVD^*_1(S')=\VD(S')$. 
    The quantities $\sum_{c} \bar{v}_{c,1}(S')$ and $\sum_c \bar{u}_{c,1}(S')$ count the vertices 
    and the unbounded edges, respectively, 
    of the farthest  Voronoi diagram~$\mCVD^*_1(S') = \FVD(S')$. 
    Every region of~$\VD(S')$ is nonempty and simply connected~\cite{k-cavd-89}
    and every vertex is of degree three by the general position assumption.
    Euler's formula thus implies the first equation.
    Every region of~$\FVD(S')$ is unbounded and $\FVD(S')$ forms a tree~\cite{mmr-fsavd-01,bcklpz-choavd-15}.
    So, Euler's formula together with the general position assumption
    implies the second equation.
\end{proof}


\begin{restatable}{lemma}{complexity} \label{lem:complexity}
 For each~$1 \leq k \leq m-1$, let
 $\SumU_k := \sum_{1\leq j \leq k} (u_{2,j} + (k-j+1) u_{1,j})$
 and $\mSumU_k := \sum_{1\leq j \leq k} (\bar{u}_{2,j} + (k-j+1) \bar{u}_{1,j})$.
	The total number of vertices in $\CVD_k(S)$ is exactly
	\[ 2k(2n-k)-2n - \SumU_k - \SumU_{k-1}
	- 2\sum_{j=1}^{k-1} v_{2,j} - \sum_{j=1}^{k} (2k-2j+1)v_{1,j} \]
	and the number of vertices in $\mCVD_k(S)$ is exactly
	\[ \mSumU_k + \mSumU_{k-1} - 2k^2
	- 2\sum_{j=1}^{k-1}\bar{v}_{2,j}-\sum_{j=1}^{k}(2k-2j+1)\bar{v}_{1,j}.\]
\end{restatable}
\begin{proof}
    Following~\cite{bop-hocvdccsf-25}, we first build two ordinary CS-structures
    for Voronoi vertices.
    
    Let $\mathcal{V}= \mathcal{V}(S)$ be the set
    of all 
    Voronoi vertices created by any triplet of sites, i.e.,
    vertices of~$\VD(\{p,q,r\})$ for any three distinct sites~$p,q,r\in S$,
    regardless of colors. Recall that there are at most two possible Voronoi vertices for each triplet of
    sites~\cite{k-cavd-89}. 
    The set $\mathcal{V}$ is the set of configurations.
    Consider two conflict relations~$\chi, \bar{\chi} \subseteq S \times \mathcal{V}$
    defined as follows: for each~$p\in S$ and $v\in \mathcal{V}$ whose defining sites are~$q_1,q_2, q_3\in S$,
    \[
    (p,v) \in \chi~\text{iff}~p \sprec_v q_1 \sequiv_v q_2 \sequiv_v q_3
    \quad\text{and}\quad
    (p,v) \in \bar{\chi}~\text{iff}~q_1 \sequiv_v q_2 \sequiv_v q_3 \sprec_v p.
    \]
    The set systems~$(S, \mathcal{V}, \chi)$ and $(S, \mathcal{V}, \bar{\chi})$
    are CS-structures.
    
    We then consider the color assignment~$\kappa \colon S \to K$ and derive two color-augmented CS-structures.
    Let $\chi_\kappa, \bar{\chi}_\kappa \subseteq K \times \mathcal{V}$ 
    be the corresponding  color 
    conflict relations with the following property:
    for each color~$a\in K$ and~$v\in\mathcal{V}$ defined by three sites~$q_1, q_2, q_3 \in S$,
    \[
    (a, v) \in \chi_\kappa~\text{iff}~a \colprec_v 
    \kappa(q_1) \colequiv_v \kappa(q_2) \colequiv_v \kappa(q_3)
    \;\text{and}\;
    (a, v) \in \bar{\chi}_\kappa~\text{iff}~a \mcolprec_v
    \kappa(q_1) \mcolequiv_v \kappa(q_2) \mcolequiv_v \kappa(q_3).
    \]
    
    By Lemma~\ref{lem:vertex_types}, 
    $v$ is a new $c$-chromatic vertex of~$\CVD^*_{k}(S)$
    if and only if
    $v$ is a colored configuration induced by~$(S, \mathcal{V}, \chi)$
    such that $v$ is in conflict with exactly $k-1$~colors and
    $c$ is equal to the number of distinct colors of the three
    sites defining~$v$.
    Analogously,
    $v$ is a new $c$-chromatic vertex of~$\mCVD^*_{k}(S)$
    if and only if
    $v$ is a colored configuration induced by~$(S, \mathcal{V}, \bar{\chi})$
    such that $v$ is in conflict with exactly $k-1$~colors and
    $c$ is equal to the number of distinct colors of three sites defining~$v$.
    Hence, the two color-augmented structures describe
    the new vertices of 
    $\CVD^*_k(S)$ and $\mCVD^*_k(S)$.
    
    Similarly, we can build two CS-structures encoding  the new unbounded edges
    of~$\CVD^*_k(S)$ and of~$\mCVD^*_k(S)$. 
    Recall the  closed simple curve~$\Gamma$ enclosing $\CVD^*_k(S)$
    and~$\mCVD^*_k(S)$, defined in Section~\ref{subsec:perm}.
    Let $\mathcal{U} = \mathcal{U}(S)$ be the set of all intersection points
    between the bisectors in 
    $\Bisectors$ and $\Gamma$.
    In this case, the set $\mathcal{U}$ is the set of configurations.
    Each configuration $u\in\mathcal{U}$ is defined by a unique pair of sites.

    Consider two conflict relations $\chi', \bar{\chi}' \subseteq S \times \mathcal{U}$:
   for each~$p\in S$ and $u\in \mathcal{U}$ defined by $q_1, q_2\in S$, 
    \[
    (p,u) \in \chi'~\text{iff}~p \sprec_u q_1 \sequiv_u q_2
    \quad\text{and}\quad
    (p,u) \in \bar{\chi}'~\text{iff}~q_1 \sequiv_u q_2 \sprec_u p.
    \]
    The set systems defined by the triplets~$(S, \mathcal{U}, \chi')$ and $(S, \mathcal{U}, \bar{\chi}')$
    form two CS-structures.
    The corresponding colored configurations induced by the color assignment~$\kappa$ for~$S$
    describe the new $c$-chromatic unbounded edges in~$\CVD^*_k(S)$ and~$\mCVD^*_k(S)$, respectively.
    Specifically,
    let $\chi'_\kappa, \bar{\chi}'_\kappa$ be the color-conflict relations 
    corresponding to $\chi'$ and  $\bar{\chi}'$, respectively.
    For each~$a\in K$ and~$u\in\mathcal{U}$ defined by two sites~$q_1, q_2 \in S$,
    \[
    (a, u) \in \chi'_\kappa~\text{iff}~a \colprec_u \kappa(q_1) \colequiv_u \kappa(q_2) 
    \quad\text{and}\quad
    (a, u) \in \bar{\chi}'_\kappa~\text{iff}~a \mcolprec_u \kappa(q_1) \mcolequiv_u \kappa(q_2).
    \]
    By 
    Lemma~\ref{lem:edge_types},
    we observe that a  $c$-chromatic unbounded edge $e$ in~$\CVD^*_{k}(S)$ is
    new 
    if and only if the configuration 
    $u = e \cap \Gamma$ 
    is in conflict with exactly $k-1$~colors and
    $c$ is equal to the number of distinct colors of the two sites defining~$u$.
    Analogously,
    $u = e \cap \Gamma$  is a new $c$-chromatic unbounded edge of~$\mCVD^*_{k}(S)$
    if and only if
    $u$ is in conflict with exactly $k-1$~colors and
    $c$ is equal to the number of distinct colors of the two sites defining~$u$.
    Hence, the two 
    color-augmented structures describe
    the new unbounded edges of the refined diagrams~$\CVD^*_k(S)$ and $\mCVD^*_k(S)$.
    
    
    From the lower bound lemma~\cite[Lemma~2]{bop-hocvdccsf-25} (see also~\cite{cs-arscgII-89}),
    we have the following:
    For any integer~$1\leq r \leq m$ and a random subset~$R\subseteq K$ of $r$~colors,
    \[ \binom{m}{r} \E[v_{c,1}(S_R)] = \sum_{j=0}^{m-c} v_{c,j+1}\binom{m-c-j}{r-c}
    \,\text{and}\,
    \binom{m}{r} \E[\bar{v}_{c,1}(S_R)] = \sum_{j=0}^{m-c} \bar{v}_{c,j+1}\binom{m-c-j}{r-c}
    \]
    for each $c \in \{1,2,3\}$, where $S_R = \bigcup_{a\in R} S_a$,
    and
    \[ \! \binom{m}{r} \E[u_{c,1}(S_R)] = \sum_{j=0}^{m-c} u_{c,j+1}\binom{m-c-j}{r-c}
    \,\text{and}\,
    \binom{m}{r} \E[\bar{u}_{c,1}(S_R)] = \sum_{j=0}^{m-c} \bar{u}_{c,j+1}\binom{m-c-j}{r-c}
    \]
    for each $c\in\{1,2\}$.
    Hence, on one hand, we have
    \[ \binom{m}{r}\left(\sum_{c=1}^3 \E[v_{c,1}(S_R)] + \sum_{c=1}^2 \E[u_{c,1}(S_R)]\right)
    = \sum_{c=1}^3 \sum_{j=0}^{m-c} (v_{c,j+1}+u_{c,j+1})\binom{m-c-j}{r-c}\]
    and
    \[ \binom{m}{r}\left(\sum_{c=1}^3 \E[\bar{v}_{c,1}(S_R)] - \sum_{c=1}^2 \E[\bar{u}_{c,1}(S_R)]\right)
    = \sum_{c=1}^3 \sum_{j=0}^{m-c} (\bar{v}_{c,j+1}-\bar{u}_{c,j+1})\binom{m-c-j}{r-c},\]
    where we define $u_{3,j} = \bar{u}_{3,j} = 0$ for all $j$.
    
    On the other hand, 
    the upper bound lemma~\cite[Lemma~3]{bop-hocvdccsf-25}, together with 
    Lemma~\ref{lem:v0_u0}, implies that
    \begin{align*}
        \binom{m}{r}\left(\sum_{c=1}^3 \E[v_{c,1}(S_R)] + \sum_{c=1}^2 \E[u_{c,1}(S_R)]\right)
        &= \sum_{\substack{R'\subseteq K\\ |R'| = r}}\sum_{c=1}^3 (v_{c,1}(S_{R'}) + u_{c,1}(S_{R'})) \\
        &= \sum_{R'} (2|S_{R'}| - 2) \\
        & = 2\binom{m-1}{r-1}n - 2\binom{m}{r},
    \end{align*}
    for any $2\leq r\leq m$. 
    Similarly, we also have
    \[
    \binom{m}{r}\left(\sum_{c=1}^3 \E[\bar{v}_{c,1}(S_R)] - \sum_{c=1}^2 \E[\bar{u}_{c,1}(S_R)]\right)
    = - 2\binom{m}{r},
    \]
    for any $2\leq r\leq m$.
    
    These two sets of equations for~$2\leq r\leq m$
    form two systems of linear equations, which are analogous 
    to those in~\cite[Lemma~13]{bop-hocvdccsf-25}.
    So, letting
    \begin{align*}
        \SumV_k := v_{3,k} + \sum_{j=1}^k (v_{2,j} + (k-j+1) v_{1,j}), \qquad
        & \mSumV_k := \bar{v}_{3,k} + \sum_{j=1}^k (\bar{v}_{2,j} + (k-j+1) \bar{v}_{1,j}), \\
        \SumU_k := \sum_{j=1}^k (u_{2,j} + (k-j+1) u_{1,j}), \quad\text{and}\quad
        & \mSumU_k := \sum_{j=1}^k (\bar{u}_{2,j} + (k-j+1) \bar{u}_{1,j}),
    \end{align*}
    we obtain
    \[ 
    \SumV_k + \SumU_k = k(2n-k-1) \quad\text{and}\quad \mSumV_k - \mSumU_k = -k(k+1),
    \]
    for each~$1\leq j \leq m-1$.
    From Lemma~\ref{lem:consecutive_orders}, 
    for each~$1\leq k\leq m-1$,
    we know that the numbers of vertices of~$\CVD_k(S)$ and of~$\mCVD_k(S)$ are equal to
    \[ v_{3,k} + v_{3,k-1} + v_{2,k} \quad\text{and}\quad 
    \bar{v}_{3,k} + \bar{v}_{3,k-1} + \bar{v}_{2,k},\]
    respectively, where  $v_{3, 0} = \bar{v}_{3, 0} = 0$.
    Therefore, the claimed equations are derived analogously to~\cite[Theorem~14]{bop-hocvdccsf-25}.
\end{proof}

\deleted{
Note that the equations derived in Lemma~\ref{lem:complexity} are identical to
those in \cite{bop-hocvdccsf-25}.
The proof closely follows the colorful Clarkson--Shor framework, defining
appropriate CS-structures for the order-$k$ abstract color Voronoi diagrams.
How to apply the Clarkson--Shor technique to ordinary order-$k$ abstract Voronoi diagrams was an open question in \cite{bcklpz-choavd-15}. Our result, for the uncolored case ($m=n$), settles this question.
}

Lemma~\ref{lem:complexity}
reduces the problem of bounding
the  combinatorial complexity of~$\CVD_k(S)$ and~$\mCVD_k(S)$
to the problem of bounding two quantities~$\SumU_k$ and $\mSumU_k$
regarding the number of new unbounded edges in the refined diagrams of orders~$1$
to~$k$.
These equations were also derived in \cite{bop-hocvdccsf-25}.
However, in \cite{bop-hocvdccsf-25} the two quantities
$\SumU_k$ and $\mSumU_k$ were assumed equal, which holds in the case of  point sites
under convex distances. 
However, equality 
does not hold in general for other types of sites or metrics.
In this paper  we remove this assumption
by establishing the following lemma 
whose proof is deferred to  Section~\ref{sec:SumU}.

\begin{restatable}{lemma}{sumU}\label{lem:sumU}
 For each~$1\leq k\leq m-1$, 
 $\SumU_k \geq k(k+1)$ and $\mSumU_k \leq k(2n-k-1)$.
 Both bounds are tight.
\end{restatable}


By combining Lemma~\ref{lem:complexity} and Lemma~\ref{lem:sumU} we obtain the following
combinatorial result.

\begin{restatable}{theorem}{thmcomplexity} \label{thm:complexity}
 Let $S$ be a set of $n$~colored sites with $m\leq n$~colors.
 Given an admissible system~$\Bisectors$ of bisecting curves for~$S$ 
 and an integer~$k$ with~$1\leq k\leq m-1$,
 the total number of vertices in the order-$k$ abstract color Voronoi diagram~$\CVD_k(S)$ or~$\mCVD_k(S)$ 
 is at most $4k(n-k)-2n$.
 Moreover, this bound is tight.
\end{restatable}
\begin{proof}
    By Lemma~\ref{lem:complexity},
    the number of vertices in~$\CVD_k(S)$ is 
    \[
    2k(2n-k) - 2n - \SumU_k - \SumU_{k-1} - 2\sum_{j=1}^{k-1} v_{2,j} - \sum_{j=1}^k(2k-2j+1)v_{1,j}
    \leq 4k(n-k) - 2n
    \]
    since $\SumU_k \geq k(k+1)$ as shown in Lemma~\ref{lem:sumU}.
    The number of vertices in~$\mCVD_k(S)$ is
    \[
    \mSumU_k + \mSumU_{k-1} - 2k^2 - 2\sum_{j=1}^{k-1} v_{2,j} - \sum_{j=1}^k(2k-2j+1)v_{1,j}
    \leq 4k(n-k) - 2n
    \]
    since $\mSumU_k \leq k(2n-k-1)$ as shown in Lemma~\ref{lem:sumU}.
    
    The tightness of the upper bound~$4k(n-k)-2n$ follows 
    from the tightness of the bounds in Lemma~\ref{lem:sumU}
    and the fact that the number of $c$-chromatic vertices, $c<3$,  is zero
    if $n=m$. 
\end{proof}

The bound extends also (asymptotically) to the refined diagrams  $\CVD^*_k(S)$
and $\mCVD^*_k(S)$ as  $O(k(n-k))$,  following Lemma~\ref{lem:overlay_CVDk_star} and the fact that the
monochromatic vertices are in total $O(n)$.

\deleted{
Further, for the case of~$k=m-1$, Theorem~\ref{thm:complexity} implies the following.
\begin{corollary} \label{coro:fcvd_hvd}
 The abstract farthest color Voronoi diagram and the abstract Hausdorff Voronoi diagram of~$S$
 have combinatorial complexity~$O(m(n-m+1))$.
\end{corollary}
}

Although the bound of Theorem~\ref{thm:complexity} is tight,
$\FCVD(S)$ and $\HVD(S)$ may as well have linear complexity in special cases.
In such cases  the $O(k(n-k))$ bound is not tight for large $k$.
We obtain a sharper asymptotic bound  
for $m-k \in O(\sqrt{m})$ in the
following theorem. 

\evanthia{I would like to have a look in the following proof when getting a
  chance.}

\begin{restatable}{theorem}{lastorderscomplexity} \label{thm:last_orders_complexity}
    Given a set~$S$ of $n$~colored sites with $m\leq n$~colors
    and an admissible system~$\Bisectors$ of bisecting curves for~$S$,
    if $\FCVD(S')$ (resp. $\HVD(S')$) has linear complexity for any $S' \subseteq S$, then for each $\lceil \frac{2m}{3} \rceil +1 \leq k\leq m-1$ the total combinatorial complexity of $\CVD_{k}^*(S), \ldots, \CVD_{m}^*(S)$ (resp. $\mCVD_{k}^*(S), \ldots, \mCVD_{m}^*(S)$) is $O((m-k)^2 n)$.
\end{restatable}
\begin{proof}
    Recall the 
    CS-structure $(S, \mathcal{V}, \chi)$ from the proof of
    Lemma~\ref{lem:complexity}, where the set of configurations 
    $\mathcal{V}$ is the set of
    all possible Voronoi vertices,
    and the conflict relation~$\chi \subseteq S \times
    \mathcal{V}$ is defined such that a vertex $v\in \mathcal{V}$ 
    is in conflict with a site $p \in S$ if and only if
    $v$ is nearer to $p$ than to the three sites defining $v$.
    Recall also the color-augmented structure $(S, \mathcal{V},
    \chi_\kappa)$, where $\chi_\kappa \subseteq K \times \mathcal{V}$ 
    is defined such that a $c$-chromatic vertex $v$
    is in conflict with a color $a \in K$ if and only if
    $v$ is nearer to $a$, with respect to the minimal dominance~$\colpreceq_v$, than to the $c$ colors
    of the sites defining $v$.
    
    We define a new color conflict relation $\chi''_\kappa \subseteq K \times \mathcal{V}$
    that is appropriate for large values of $k$.
    For each color $a \in K$ and vertex $v \in \mathcal{V}$ defined by three sites $q_1, q_2, q_3 \in S$:
    \[ (a, v) \in \chi''_\kappa~\text{iff}~\kappa(q_1) \colequiv_v \kappa(q_2) \colequiv_v \kappa(q_3) \colprec_v a.
    \]
    We call $\{\kappa(q_1),\kappa(q_2),\kappa(q_3)\}$ the \emph{defining colors} of the color configuration $v$.
    
    Let $r = m-k+1$. 
    The number of conflicts of a color configuration $v$ with respect to
    $\chi''_\kappa$ is the number of colors farther from $v$ than its defining
    colors. A new $c$-chromatic vertex $v$ of $\CVD^*_k(S)$ has $c$ defining
    colors, and $k-1$ other colors closer than the defining colors; therefore,
    $v$ has $m-c-(k-1) = r-c$ colors farther from its defining colors. 
    
    Let $v''_{c,r-c+1}$ denote the number of $c$-chromatic vertices that are in
    conflict with exactly $r-c$ colors with respect to $\chi''_\kappa$. Note
    that $v_{c,k} = v''_{c,r-c+1}$. From Lemma~\ref{lem:consecutive_orders}, it
    follows that the total number of vertices in $\CVD^*_k(S)$ is equal to 
    \begin{alignat*}{6}
        & v_{1,k} && + v_{2,k} && + v_{2,k-1} && + v_{3,k} && + v_{3,k-1} && + v_{3,k-2} \\
        = \quad & v''_{1,r} && + v''_{2,r-1} && + v''_{2,r} && + v''_{3,r-2} && + v''_{3,r-1} && + v''_{3,r}
    \end{alignat*}
    where $v''_{3,-1} = v''_{3,0} = v''_{2,0} = 0$.
    
    Assuming that $\FCVD(S)$ has linear complexity, 
    it follows that $\sum_c v''_{c,1} = O(n)$. 
    Note that for $c=1$, the total number of monochromatic vertices is $\sum_{r} v''_{1,r} = O(n)$.
    Theorem~4 in \cite{bop-hocvdccsf-25} applies to any color-augmented CS-structure, so it also applies to $(S,\mathcal{V},\chi''_\kappa)$.
    Let $|\conf_{c,j}(S, \chi''_\kappa)|$ denote the number of $c$-chromatic configurations that are in conflict with exactly $j$ colors; observe that $v''_{c,j+1} = |\conf_{c,j}(S, \chi''_\kappa)|$.
    Applying the theorem with the 
    function $T_0(x)=x$ (since the complexity is linear for $j=0$ and any $S' \subseteq S$), for each $c \in \{2,3\}$ and $1 \leq r \leq \lfloor m/c \rfloor$, the number of $c$-chromatic vertices with at most $r-1$ conflicts is
    \[ \sum_{j=0}^{r-1} v''_{c,j+1} = \sum_{j=0}^{r-1} |\conf_{c,j}(S, \chi''_\kappa)|
    = O\left(\frac{r^c}{m} \cdot T_0\left(\frac{mn}{r}\right)\right)
    = O\left(\frac{r^c}{m} \cdot \frac{mn}{r}\right)
    = O(r^{c-1} n).
    \]
    Thus, the 
    sum $\sum_{j=0}^{r-1} v''_{c,j+1}$ for all $c\in\{1,2,3\}$ is bounded by $O(n + r n + r^2 n)$.
    Substituting $r = m - k + 1$, the combinatorial complexity of
    $\CVD_{k}^*(S), \ldots, \CVD_{m}^*(S)$ for $\lceil \frac{2m}{3} \rceil +1
    \leq k\leq m-1$ is $O((m-k)^2 n)$. 
    
    Similarly, we define the color conflict relation $\bar{\chi}''_\kappa \subseteq K \times \mathcal{V}$
    for each color $a \in K$ and vertex $v \in \mathcal{V}$ defined by three sites $q_1, q_2, q_3 \in S$:
    \[ (a, v) \in \bar{\chi}''_\kappa~\text{iff}~\kappa(q_1) \mcolequiv_v \kappa(q_2) \mcolequiv_v \kappa(q_3) \mcolprec_v a.
    \]
    The bound on $\mCVD_{k}^*(S), \ldots, \mCVD_{m}^*(S)$ follows analogously.
\end{proof}



A concrete example is the order-$k$ polygon Voronoi diagram of  
$m$ disjoint simple polygons of total complexity $n$, where
the farthest polygon Voronoi diagram has complexity $O(n)$~\cite{cegghlln-fpvd-11}.


\begin{corollary}
  Given a collection of $m$ disjoint simple polygons of $n$ total vertices and an integer~$k$ with~$1\leq
  k\leq m-1$, the order-$k$ polygon Voronoi diagram has combinatorial
  complexity $O(\min\{k(n-k),(m-k)^2n\})$.
\end{corollary}

\section{Tight bounds on \texorpdfstring{$\SumU_k$}{Uk} and
  \texorpdfstring{$\mSumU_k$}{mUk} -- Proving Lemma~\ref{lem:sumU}}
\label{sec:SumU}

Recall the circular sequence~$\Pi(S) = (\pi_0, \ldots, \pi_{N-1}, \pi_N=\pi_0)$
of permutations~$\pi_i$ of sites~$S$, discussed in Section~\ref{subsec:perm}.
In our setting, each site~$p\in S$ is assigned a color~$\kappa(p) \in K$.
For each color~$a\in K$ and $0\leq i\leq N-1$,
let $\lambda_i(a) \in S_a$ be the first site of color~$a$ in permutation~$\pi_i$,
called the \emph{leader} of color~$a$ in~$\pi_i$.
By a slight abuse of notation,
we also  use~$\lambda_i$ to denote
the subsequence of~$\pi_i$ that consists of the leaders of all colors.
It is not hard  to see  that 
the induced circular sequence
$\Lambda(\Pi(S)) = (\lambda_0, \ldots, \lambda_N = \lambda_0)$ 
of color-leader sequences satisfies exactly one of the following, if
$\lambda_{i+1} \neq \lambda_i$: 
\begin{itemize} 
 \item \hypertarget{def:lpermL1}{(L1)} there is a replacement of the leader of a color~$a\in K$,
  that is, $\lambda_{i+1}(a) \neq \lambda_i(a)$, or
 \item \hypertarget{def:lpermL2}{(L2)} there is a switch of two consecutive leaders from~$\lambda_i$ to~$\lambda_{i+1}$,
\end{itemize}
A change of the first kind~\hyperlink{def:lpermL1}{(L1)} is called a \emph{monochromatic switch}
as it corresponds to a switch in~$\Pi(S)$ between two sites of the same color;
a change of the second kind~\hyperlink{def:lpermL2}{(L2)} is called a \emph{bichromatic switch}.
We say that a monochromatic switch happens \emph{at position~$j$},
if the leader~$\lambda_i(a)$ of a color~$a\in K$ is at position $j$ 
in~$\lambda_i$,
and that a bichromatic switch happens \emph{at positions~$j$},
if two leaders $\lambda_i(a)$ and $\lambda_i(b)$ (to be switched) 
lie at positions~$j$ and $j+1$ in~$\lambda_i$. 
We then observe the following
correspondence between unbounded edges and switches.

\begin{restatable}{lemma}{UG} \label{lem:U-G}
 Let $\overline{\Pi}(S) := (\bar{\pi}_0, \ldots, \bar{\pi}_N = \bar{\pi}_0)$,
 where $\bar{\pi}_i$ denotes the reverse of permutation~$\pi_i$.
 For each~$1\leq k\leq m$, the following hold for the $c$-chromatic unbounded
 edges, $c\in\{1,2\}$:
 \begin{itemize} 
  \item $u_{c,k}$ is equal to the number of $c$-chromatic switches at position~$k$ in~$\Lambda(\Pi(S))$.
 \item $\bar{u}_{c,k}$ is equal to the number of $c$-chromatic switches at position~$k$ 
 in~$\Lambda(\overline{\Pi}(S))$.
\end{itemize}
\end{restatable}
\begin{proof}
    The proof is done by establishing one-to-one correspondences between:
    \begin{itemize} 
        \item The set of new monochromatic unbounded edges of~$\CVD^*_k(S)$ and
        the set of monochromatic switches at position~$k$ in~$\Lambda(\Pi(S))$.
        \item The set of new bichromatic unbounded edges of~$\CVD^*_k(S)$ and
        the set of bichromatic switches at position~$k$ in~$\Lambda(\Pi(S))$.
        \item The set of new monochromatic unbounded edges of~$\mCVD^*_k(S)$ and
        the set of monochromatic switches at position~$k$ in~$\Lambda(\overline{\Pi}(S))$.
        \item The set of new bichromatic unbounded edges of~$\mCVD^*_k(S)$ and
        the set of bichromatic switches at position~$k$ in~$\Lambda(\overline{\Pi}(S))$.
    \end{itemize}
    
    Let $x \in \Gamma$, and let
    $a_1 \colpreceq_x \cdots \colpreceq_x a_m$ and $p_1 \spreceq_x \cdots \spreceq_x p_n$
    be the ordering of colors in~$K$ with respect to~$\colpreceq_x$
    and the ordering of sites in~$S$ with respect to~$\spreceq_x$, respectively.
    
    Suppose that it holds that
    $a_{k-1} \colprec_x a_k \colprec_x a_{k+1}$ and
    $p_{j-1} \sprec_x p_j \sequiv_x p_{j+1} \sprec_x p_{j+2}$
    for some~$j$ such that $p_j, p_{j+1} \in S_{a_k}$, and $p_i \notin S_{a_k}$ for all $i<j$.
    Then  $p_j, p_{j+1}$
    are the nearest sites among those in~$S_{a_k}$ at~$x$.
    Thus,  $x = e\cap \Gamma$
    for a new monochromatic unbounded edge~$e$ of~$\CVD^*_k(S)$ (see also Lemma~\ref{lem:edge_types}).
    At the same time,
    $x \in \Gamma$ corresponds to the switch~$\sigma$ of~$p_j$ and~$p_{j+1}$
    between two consecutive permutations~$\pi_i,\pi_{i+1} \in \Pi(S)$.
    As both $p_j$ and $p_{j+1}$ belong to the same color~$S_{a_k}$
    and $a_k$ is the $k$-th nearest color locally near~$x$,
    the switch~$\sigma$ is monochromatic happening at position~$k$
    by which the leader of color~$a_k$ is replaced between~$\lambda_i$ and~$\lambda_{i+1}$.
    Conversely, any monochromatic switch of~$p_j$ and $p_{j+1}$ of color~$a_k$
    happening at position~$k$ corresponds to a point~$x\in \Gamma$
    that satisfies the above condition.
    This establishes the first one-to-one correspondence between
    the set of new monochromatic unbounded edges of~$\CVD^*_k(S)$ and
    the set of monochromatic switches at position~$k$ in~$\Lambda(\Pi(S))$,
    implying that the number of monochromatic switches at position~$k$ in~$\Lambda(\Pi(S))$
    is exactly~$u_{1,k}$.
    
    Next, if it holds that
    \( a_{k-1} \colprec_x a_k \colequiv_x a_{k+1} \colprec_x a_{k+2}, \)
    then $x \in \cl \CVR_k(H, S) \cap \cl \CVR_k(H',S)$ for $H = \{a_1,
    \ldots, a_k\}$ and $H' = H \cup \{a_{k+1}\} \setminus \{a_{k}\}$.
    Thus $x = e \cap \Gamma$ for a new bichromatic unbounded edge~$e$
    of~$\CVD^*_k(S)$ (see Lemma~\ref{lem:edge_types}).
At the same time,
$x$ corresponds to the switch of the leaders of two colors~$a_k$ and~$a_{k+1}$
between two consecutive sequences~$\lambda_i$ and~$\lambda_{i+1}$ of color leaders,
so a bichromatic switch at position~$k$.
Conversely, for any bichromatic switch at position~$k$ corresponding to a point~$x\in\Gamma$,
we have the same condition as above at~$x$,
so $x$ lies on a new bichromatic unbounded edge of~$\CVD^*_k(S)$.
This implies that
the number of bichromatic switches at position~$k$ in~$\Lambda(\Pi(S))$
is equal to~$u_{2,k}$.

The other two cases for the maximal diagrams~$\mCVD^*_k(S)$
can be shown analogously using the sequence~$\overline{\Pi}(S)$ of reversed permutations.
Let $\bar{\lambda}_i$ be the sequence of color leaders 
induced from the reversed permutation~$\bar{\pi}_i$ such that $\Lambda(\overline{\Pi}(S)) = 
(\bar{\lambda}_0, \ldots, \bar{\lambda}_{N-1}, \bar{\lambda}_N=\bar{\lambda}_0)$.
Notice that the $k$-th farthest color with respect to the maximal dominance~$\mcolpreceq_x$
is at position~$k$ in the corresponding color-leader sequence~$\bar{\lambda}_i$.
Hence, the characterization of new edges of~$\mCVD^*_k(S)$, analogous to Lemma~\ref{lem:edge_types},
implies the last two correspondences between 
new unbounded edges of~$\mCVD^*_k(S)$ and switches at position~$k$ in~$\Lambda(\overline{\Pi}(S))$.
Consequently, 
the number of monochromatic switches at position~$k$ in~$\Lambda(\overline{\Pi}(S))$
is equal to~$\bar{u}_{1,k}$, 
while
the number of bichromatic switches at position~$k$ in~$\Lambda(\overline{\Pi}(S))$
is equal to~$\bar{u}_{2,k}$.
\end{proof}

Thus, to obtain the bounds of Lemma~\ref{lem:sumU}, it is enough 
to count 
monochromatic and bichromatic switches
in the circular sequences~$\Lambda(\Pi(S))$ and~$\Lambda(\overline{\Pi}(S))$
of color-leaders.
%
We do
the counting in a purely combinatorial setting
in the next  subsection 
deriving the bounds of Theorem~\ref{thm:Gk}.
As a corollary to  Theorem~\ref{thm:Gk}
we get:
$\SumU_k \geq k(k+1)$ and $\mSumU_k \leq k(2n-k-1)$.
These bounds are tight for the unbounded edges of the diagrams  by the fact 
that any circular sequence~$\Pi$ of permutations 
satisfying conditions~\hyperlink{def:permP1}{(P1)} and~\hyperlink{def:permP2}{(P2)} can be realized by
an admissible system of bisecting curves 
\cite[Lemma~10]{bcklpz-choavd-15}.
This completes the proof of Lemma~\ref{lem:sumU}.

\subsection{Switches in a sequence of permutations of colored elements}
In the following, we consider a purely combinatorial setting
for circular sequences as follows.
Let $\Pi = (\pi_0, \ldots, \pi_N = \pi_0)$ be a given circular sequence
of permutations~$\pi_i$ of $n$~elements $\{1, \ldots, n\}$ that satisfies conditions~\hyperlink{def:permP1}{(P1)} and~\hyperlink{def:permP2}{(P2)},
and assume that each element in~$\{1, \ldots, n\}$ is assigned a color from~$K =\{1, \ldots, m\}$.
As above, we have the induced sequence~$\Lambda = \Lambda(\Pi) = (\lambda_0, \ldots, \lambda_{N-1},\lambda_N=\lambda_0)$
of sequences~$\lambda_i$ of $m$~leaders of colors, which satisfy conditions~\hyperlink{def:lpermL1}{(L1)} and~\hyperlink{def:lpermL2}{(L2)}.
For each~$1\leq k\leq m$, define
$\bm{g_{1,k}} = g_{1,k}(\Pi)$ to be the number of monochromatic switches at position~$k$ in~$\Lambda$
and 
$\bm{g_{2,k}} = g_{2,k}(\Pi)$ to be the number of bichromatic switches at position~$k$ in~$\Lambda$.
We are interested in the following quantity:
\[
 G_k = G_k(\Pi) := \sum_{j=1}^k ( g_{2,j} + (k-j+1) g_{1,j}).
\]
%

\begin{restatable}{theorem}{Gk} \label{thm:Gk}
 For 
 $1\leq k\leq m-1$,
 $k(k+1)\leq G_k\leq k(2n-k-1)$.
 Both bounds are tight.
\end{restatable}

The bounds 
are shown to be
tight by constructions, see Lemma~\ref{lem:Gk_tight}. 


\subsubsection{Deriving the lower bound}
\label{subsec:lower_bound}
We charge each monochromatic or bichromatic switch to a unique color~$a\in K$ as follows:
\begin{itemize} 
    \item Each monochromatic switch, replacing the leader of color~$a\in K$, 
    is charged to the corresponding color~$a$.
    Let $\bm{\mu_a(j)}$ be the number of monochromatic switches at position~$j$ in~$\Lambda$
    charged to~$a$.
    \item Each bichromatic switch involving two colors~$a,a' \in K$ with~$a< a'$
    is charged to color~$a$ with the smaller index.
    Let $\bm{\beta_a(j)}$ be the number of bichromatic switches at position~$j$ in~$\Lambda$
    charged to~$a$.
\end{itemize}

It is obvious that $g_{1,j} = \sum_{a\in K} \mu_a(j)$ and $g_{2,j} = \sum_{a\in K} \beta_a(j)$.

Let $\rho_i(a)$ for $a\in K$ be the position, or the rank, of~$\lambda_i(a)$ 
in permutation~$\lambda_i$,
and let $\rho^*(a) := \max_{i} \rho_i(a)$ be the maximum rank of~$a$ over the whole sequence~$\Lambda$.
Note that for any color~$a\in K$ having a single element,
we have $\mu_a(j) = 0$ for every~$j$, since there is no change of its leader~$\lambda_i(a)$.
We observe the following for colors having at least two elements.
\begin{restatable}{lemma}{mcineqtwosites} \label{lem:mc_ineq_2_sites}
    Let $a \in K$ be any color having at least two elements.
    For $1\leq k\leq m-1$,
    \[
    \sum_{j=1}^k (k-j+1) \mu_a(j) \geq \max \{2(k + 1 - \rho^*(a)), 0\}.
    \]
\end{restatable}
\begin{proof}
    If $k < \rho^*(a)$, then $\max \{k + 1 - \rho^*(a), 0\} = 0$ and the inequality holds trivially,
    so assume $k \geq \rho^*(a)$. 
    
    We first observe that there are at least two changes of~$\lambda_i(a)$ over $0\leq i\leq N-1$.
    This is obvious from the fact that $\lambda_i(a)$ is the first element
    of the subsequence of~$\pi_i$ that contains those elements in~$\{1, \ldots, n\}$ whose color is~$a$
    and the property~\hyperlink{def:permP2}{(P2)} that every pair of elements switches exactly twice.
    
    Consider any two such changes happened at~$0\leq i_1 < i_2 < N$, and 
    let $j_1:=\rho_{i_1}(a)$ and $j_2:=\rho_{i_2}(a)$ be two positions 
    of the leader of~$a$ in $\lambda_{i_1}$ and $\lambda_{i_2}$, respectively.
    By definition, we have 
    $j_1, j_2 \leq \rho^*(a) \leq k$.
    Thus, we obtain 
    \[
    \sum_{j=1}^k (k-j+1) \mu_a(j) \geq (k-j_1+1) + (k-j_2+1) 
    \geq 2 (k - \rho^*(a) +1)
    \]
    by the existence of two changes of~$\lambda_i(a)$.
\end{proof}


Next, we consider the quantity~$\beta_a(j)$ about bichromatic switches.
%
Note that $\beta_a(m) = 0$ for any~$a\in K$ and $\beta_m(k) = 0$ for any $k$ by definition.
From here onward, we assume without loss of generality that $\pi_0 = (1, 2, \ldots, n)$
and 
the colors in~$K =\{1, \ldots, m\}$ appear in the order in the initial permutation~$\lambda_0$
of leaders, that is,
$\lambda_0 = (\lambda_0(1), \lambda_0(2), \ldots, \lambda_0(m))$,
so $\rho_0(a) = a$ for $a\in K = \{1, \ldots, m\}$.
\begin{restatable}{lemma}{biineq} \label{lem:bi_ineq}
    For $1 \leq k \leq m-1$ and $a \in K$, it holds that
    $\sum_{j=1}^k \beta_a(j) \geq \max\{2 ( \khat_a + 1 - a ), 0\}$,
    where $\khat_a = \min\{ k, \rho^*(a)-1 \}$. 
\end{restatable}
\begin{proof}
    Recall our assumption that the leader of color~$a$ is placed at position~$a$ in~$\lambda_0$,
    that is, $\rho_0(a) = a$.
    If $\khat_a < a - 1$, then the inequality holds trivially,
    so assume that $\khat_a \geq a$.
    
    Let $z$ be an index such that $\rho_z(a) = \khat_a + 1$, and
    let 
    \[ A := \{ a'\in K \mid a' > a~\text{and}~\rho_z(a') < \rho_z(a) = \khat_a+1 \}.\]
    be the set of colors~$a' > a$ preceding~$a$ in~$\lambda_z$.
    Since the order between colors~$a$ and each~$a' \in A$ is reversed 
    between~$\lambda_0$ and~$\lambda_z$,
    there must be at least $|A|$ many bichromatic switches while we move from $\lambda_0$ to $\lambda_z$,
    counted in~$\sum_{j=1}^{\khat_a} \beta_a(j)$.
    On the other hand, while we move from~$\lambda_z$ to~$\lambda_N = \lambda_0$,
    this order between $a$ and each~$a'\in A$ should be reversed again.
    This implies that
    \[ \sum_{j=1}^k \beta_a(j) \geq \sum_{j=1}^{\khat_a} \beta_a(j) \geq 2|A|.\]
    
    To bound~$|A|$, let $B := \{ a'\in K \mid a' < a~\text{and}~\rho_z(a') < \rho_z(a)\}$.
    Obviously, $|B| \leq a-1$ and $A\cap B =\emptyset$.
    Since $A\cup B$ consists of all colors~$a'$ such that $\rho_z(a') < \rho_z(a) = \khat_a + 1$,
    we have $|A\cup B| = \khat_a$ and hence
    \[ |A| = |A\cup B| - |B| \geq \khat_a - (a - 1) = \khat_a + 1 - a,\]
    implying the inequality.
\end{proof}



For colors~$a\in K$ that have only a single element, we show a better lower
bound on~$\beta_a$ in the next lemma. By combining  we derive the inequality in 
Lemma~\ref{lem:sub-goal}. 

\begin{restatable}{lemma}{biineqonesite} \label{lem:bi_ineq_1_site}
    Let $a \in K$ be a color having only one element.
    For $1\leq k\leq m-1$, it holds that
    $\sum_{j=1}^k \beta_a(j) \geq \max\{2 ( k + 1 - a ), 0\}$.
\end{restatable}
\begin{proof}
    Let $t \in\{1, \ldots, n\}$ be the only element of color~$a$,
    so $t$ is always the leader of~$a$: $\lambda_i(a) = t$ for all~$i$.
    Since every pair of elements switch exactly twice in the sequence~$\Pi$ (condition~\hyperlink{def:permP2}{(P2)}),
    $t$ is involved in at least two bichromatic switches with each of the other colors.
    Among them, those switches with $a' > a$ are counted in~$\sum_j \beta_a(j)$,
    so we have
    \[ \sum_j \beta_a(j) = \sum_{j=1}^{\rho^*(a)} \beta_a(j) \geq 2(m-a).\]
    
    We have two cases: whether $k \geq \rho^*(a)$ or $k < \rho^*(a)$.
    If $k \geq \rho^*(a)$, then we have
    \[
    \sum_{j=1}^{k} \beta_a(j) \geq \sum_{j=1}^{\rho^*(a)} \beta_a(j)
    \geq 2 (m - a) \geq \max\{2 ( k + 1 - a ), 0\}
    \]
    since $k \leq m-1$.
    Otherwise, if $k < \rho^*(a)$, then Lemma~\ref{lem:bi_ineq} is strong enough to conclude the lemma:
    \[
    \sum_{j=1}^k \beta_a(j)
    \geq \max\{2 ( \khat_a + 1 - a ), 0\} \geq \max\{2 ( k + 1 - a ), 0\}
    \]
    since $\khat_a = \min\{k, \rho^*(a) - 1\} = k$.
\end{proof}

\begin{restatable}{lemma}{subgoal} \label{lem:sub-goal}
    For $1 \leq k \leq m-1$ and $a \in K$, it holds that
    \[
    \sum_{j=1}^{k} \beta_a(j) + \sum_{j=1}^{k} (k-j+1) \mu_a(j) \geq \max\{2(k+1-a)), 0\}.
    \]
\end{restatable}
\begin{proof}
    If there is only one element in~$\{1, \ldots, n\}$ whose color is~$a$,
    then $\mu_a(j) = 0$ for all~$j$, so Lemma~\ref{lem:bi_ineq_1_site}
    implies the lemma.
    So, in the following, we assume that there are at least two elements whose color is~$a$.
    
    Adding the inequalities of Lemmas~\ref{lem:mc_ineq_2_sites} and~\ref{lem:bi_ineq} yields
    \[
    \sum_{j=1}^{k} \beta_a(j) + \sum_{j=1}^{k} (k-j+1)\mu_a(j) 
    \geq \max\{2 (\khat_a + 1 - a), 0\} + \max\{2 ( k + 1 - \rho^*(a) ), 0\}.
    \]
    
    We simplify the right-hand side by a case analysis on $\khat = \min\{k, \rho^*(a)-1\}$.
    If $\rho^*(a) \leq k$, then
    $\khat_a = \rho^*(a)-1$ and $\khat_a + 1 - a = \rho^*(a) - a$. 
    Note also the trivial fact that $\rho^*(a) - a = \rho^*(a) - \rho_0(a) \geq 0$.
    So, the right-hand side is simplified to
    \[
    2(\rho^*(a)-a) + \max\{2(k-\rho^*(a)+1),0\} = \max\{2(k + 1 - a), 0\}.
    \]
    Otherwise, if $\rho^*(a) > k$, then $\khat_a = k$ and $\max\{2 ( k + 1 - \rho^*(a)), 0\} = 0$.
    So, 	we have
    \[
    \max\{2 (\khat_a + 1 - a), 0\} + \max\{2 ( k + 1 - \rho^*(a) ), 0\} = \max\{2 ( k + 1 - a ), 0\},
    \]
    as claimed.
\end{proof}

Summing the inequality of Lemma~\ref{lem:sub-goal} over all colors~$a\in K$,
we obtain the claimed lower bound of Theorem~\ref{thm:Gk},
as $\sum_{a\in K} \mu_a(j) = g_{1,j}$ and $\sum_{a\in K} \beta_a(j) = g_{2,j}$.
Specifically, we have 
\[
	G_k = \sum_{a\in K}\sum_{j=1}^k (\beta_a(j) + (k-j+1)\mu_a(j))
	\geq \sum_{a=1}^m \max\{2(k+1-a), 0\} = 2 \sum_{a=1}^k a = k(k+1).
\]

\subsubsection{Deriving the upper bound}

Recall our assumption that $\pi_0 = (1,2, \ldots, n)$
and $\lambda_0 = (\lambda_0(1), \ldots, \lambda_0(m))$ as above.
We start with the following observation.

\begin{restatable}{lemma}{mslast} \label{lem:ms_last}
    For 
    $a\in K$, the largest element~$t$ of color~$a$ can be
    the leader of~$a$
    at most once in~$\pi_i$ over all~$0\leq i\leq N$.
    Thus, $t$ 
    is involved in exactly zero or two monochromatic switches.
\end{restatable}
\begin{proof}
If $t$ is the only element of color~$a$,
then the lemma is trivially true, since $\lambda_i(a) = t$ for all~$i$
and thus there is no monochromatic switches involving~$t$.
So, assume there is at least one more element~$t' < t$ of the same color~$a$.
Let $T_a$ be the set of elements~$t'$ whose colors are commonly~$a$.
Consider the subsequence~$\tau_i$ of~$\pi_i$ that contains the elements in~$T_a$.
Observe that $t$ is at the first position in~$\tau_i$
if and only if $t = \lambda_i(a)$.

Now, suppose that $t = \lambda_i(a)$ for some~$i$.
Since $t$ is at the last position in the initial permutation~$\tau_0$,
in order for~$t$ to step forward to the first position,
every other member in~$T_a$ has to be switched with~$t$,
and then they have to be switched with~$t$ again
by condition~\hyperlink{def:permP2}{(P2)}
which requires that
every pair of elements is switched exactly twice in~$\Pi$.
If the second switch between $t$ and any~$t' \in T$ happens in~$\tau_z$ at some~$0\leq z\leq N-1$,
then $t'$ precedes~$t$ in~$\tau_i$ for any~$z< i \leq N-1$.
This proves the lemma.
\end{proof}

To prove the upper bound of~$G_k$, we introduce another charging scheme,
in which every switch of our interest is charged to a unique element~$t\in\{1, \ldots, n\}$
as follows.
\begin{itemize} 
    \item Every monochromatic switch in~$\Lambda$ 
    induced by a switch of elements~$t$ and~$t'$ with~$1\leq t' < t \leq n$
    in~$\Pi$ is charged to element~$t$ with the larger index.
    Let $\bm{\hat{\mu}_t(j)}$ be the number of monochromatic switches at position~$j$ 
    in~$\Lambda$ charged to element~$t$.
    \item Every bichromatic switch in~$\Lambda$
    induced by a switch of elements~~$t$ and~$t'$ with~$1\leq t' < t \leq n$ in~$\Pi$
    is charged to element~$t$ with the larger index.
    Let $\bm{\hat{\beta}_t(j)}$ be the number of bichromatic switches 
    at position~$j$ in~$\lambda$ charged to element~$t$.
\end{itemize}
In short, each monochromatic and bichromatic switch is charged to
the involved element~$t$ of the larger index.
Obviously, we have
$g_{1,j} = \sum_{1\leq t\leq n} \hat{\mu}_t(j)$ and
$g_{2,j} = \sum_{1\leq t\leq n} \hat{\beta}_t(j)$.

The following is an immediate implication of Lemma~\ref{lem:ms_last}.
\begin{restatable}{lemma}{barmc} \label{lem:bar_mc}
For $1\leq t\leq n$,
$\sum_{j=1}^{m} \hat{\mu}_t(j) \leq 2$;
if $t = \lambda_0(a)$ for~$a\in K$, then $\sum_{j=1}^{m} \hat{\mu}_t(j) = 0$.
\end{restatable}
\begin{proof}
If $t = \lambda_0(a)$ is the initial leader of color~$a$,
then there is no other element~$t' < t$ of the same color~$a$.
Hence, there is no monochromatic switch contributed to~$\hat{\mu}_t(j)$ for all~$j$.
So, we assume that
$t$ is in color~$a$ and
there is at least one more element~$t' < t$ of the same color~$a$.
Here, we redefine $T_a$ to be the set of elements~$t' \leq t$ whose colors are commonly~$a$,
and 
consider the subsequence~$\tau_i$ of~$\pi_i$ that contains the elements in~$T_a$,
as in the proof of Lemma~\ref{lem:ms_last}.
Observe that any monochromatic switch in~$\Lambda$ involving~$t$
corresponds to a switch between the first and second positions that involves~$t$ in the~$\tau_i$'s,
while the converse is not necessarily true.
Hence, we have $\sum_{j=1}^{m} \hat{\mu}_t(j) \leq 2$ by Lemma~\ref{lem:ms_last}.
\end{proof}

\begin{restatable}{lemma}{barsubgoal} \label{lem:bar_sub-goal}
For $1 \leq k \leq m-1$ and $1\leq t\leq n$, it holds that
\[
\sum_{j=1}^{k} (\hat{\beta}_t(j) + (k-j+1) \hat{\mu}_t(j)) \leq 
\begin{cases} 
    2 \min\{a-1, k\} & \text{if}~t=\lambda_0(a)~\text{for}~a\in K\\
    2k & \text{otherwise}
\end{cases}
\]
\end{restatable}
\begin{proof}
First, assume that $t=\lambda_0(a)$ for some color~$a\in K$.
Note, by Lemma~\ref{lem:bar_mc}, that $\sum_j \hat{\mu}_t(j)=0$,
so we are interested in an upper bound of~$\sum_{j=1}^k \hat{\beta}_t(j)$.
Let $T = \{1, \ldots, t\}$ be the set of elements preceding~$t$ in~$\pi_0$, including~$t$.
We consider another circular sequence $\Pi' = (\pi'_0, \ldots, \pi'_{N-1})$
of permutations~$\pi'_i$ of~$T$ 
such that $\pi'_i$ is the subsequence of~$\pi_i$
that contains the elements in~$T$.
We also let $\Lambda' = \Lambda(\Pi') = (\lambda'_0, \ldots, \lambda'_{N-1})$
be the sequence of color leaders induced from~$\Pi'$.
Note that each bichromatic switch in~$\Lambda$ counted in the sum~$\sum_{j} \hat{\beta}_t(j)$
corresponds to a bichromatic switch in~$\Lambda'$ involving~$t$,
since $t$ is the largest element in~$T$ and 
we charge each switch in~$\Lambda$ to the larger element of the two involved elements.
Hence,
it suffices for our purpose 
to have an upper bound on the number of bichromatic switches involving~$t$ in~$\Lambda'$.

Let $T_b \subseteq T$ for~$b \in K$ be the set of elements in~$T$ whose color is~$b$.
Note that $T_a = \{t\}$ and $T_b = \emptyset$ for $b > a$ by our construction.
Pick any color~$b < a$.
While we move from~$\pi'_0$ to~$\pi'_{N-1}$,
observe that the first bichromatic switch between~$t$ and the leader of color~$b$
can happen after switches between~$t$ and every member~$t'\in T_b$,
since $t$ is at the last position in~$\pi'_0$.
If there is the second bichromatic switch between~$t$ and the leader~$t'$ of~$b$,
then $t'$ precedes~$t$ in~$\pi'_i$ afterwards until~$i=N$
by condition~\hyperlink{def:permP2}{(P2)},
so $t$ will never make a third bichromatic switch with the leader of color~$b$.
Therefore, for each color~$b<a$, $t$ is involved in at most two bichromatic switches with
the leader of~$b$, implying that
\[ 
\sum_{j=1}^{m-1} \hat{\beta}_t(j) \leq 2(a-1).
\]

Now, we consider those switches at position~$k+1$ or larger, that is,
those counted in~$\sum_{j=k+1}^{m-1}\hat{\beta}_t(j)$.
If $k \geq a$, then it is zero and we get 
$\sum_{j=1}^k \hat{\beta}_t(j) \leq 2(a-1)$.
We then assume $k < a$.
Since the initial position of~$t$ in~$\lambda'_0$ is~$a$,
at least $a - k - 1$ bichromatic switches are necessary
for~$t$ to arrive at position~$k$ or smaller.
Hence, if $\hat{\beta}_t(j) \geq 1$ for some~$j \leq k$,
then we have
\[ \sum_{j=k+1}^{m-1} \hat{\beta}_t(j) \geq 2(a-k-1)\]
since we need additional $a-k-1$ switches for~$t$ 
to return back to the original last position in~$\pi'_N = \pi'_0$.
This implies that
\[ \sum_{j=1}^k \hat{\beta}_t(j) = \sum_{j=1}^{m-1} \hat{\beta}_t(j) - \sum_{j=k+1}^{m-1}\hat{\beta}_t(j)
\leq 2(a-1) - 2(a-k-1) = 2k.
\]
Hence, we have
\[
\sum_{j=1}^k \hat{\beta}_t(j) \leq \min\{2(a-1), 2k\}
\]
if $t = \lambda_0(a)$ for~$a\in K$.

Next, 
assume that $t$ is not the leader of~$a$ in~$\pi'_0$, so $|T_a| \geq 2$. 
Since $t$ is the largest element of~$T_a$,
the total number of monochromatic switches involving~$t$ in~$\Lambda'$
is either $0$ or $2$ by Lemma~\ref{lem:ms_last}.
In the former case, $\lambda'_i(a) \neq t$ for all~$i$,
so $\hat{\beta}_t(j) = 0$ for all~$j$.
We thus assume the latter case. 
Observe that the above arguments apply to show the same upper bound
$\sum_{j=1}^k \hat{\beta}_t(j) \leq \min\{2(a-1), 2k\}$ in this case as well.
In the following, we consider the two monochromatic switches involving~$t$
to obtain the claimed upper bound.

Let $z_1, z_2$ with $0\leq z_1 < z_2 < N$ and
$j_1, j_2$ with $1\leq j_1, j_2 \leq a$ be integers such that
the first monochromatic switch involving~$t$ happens at position~$j_1$ of~$\lambda'_{z_1}$
and the second one happens at position~$j_2$ of~$\lambda'_{z_2}$.
That is, $t = \lambda'_i(a)$ for~$z_1\leq i < z_2$.
While we move from~$\pi'_0$ to~$\pi'_{z_1}$,
element~$t$ switches at least $a-j_1$ leaders of other colors~$b<a$,
but none of them are bichromatic switches in~$\Lambda'$,
as $t \neq \lambda_i(a)$ for $0\leq i<z_1$.
Analogously, while we move from~$\pi'_{z_2}$ to~$\pi'_N = \pi'_0$,
there are at least~$a-j_2$ switches that are not counted as bichromatic switches in~$\Lambda'$.
Hence, we have
\[
\sum_{j=1}^{m-1} \hat{\beta}_t(j) \leq 2(a-1) - (a-j_1) - (a-j_2) = j_1 + j_2 - 2.
\]
By a similar argument as above, 
if $j_1 > k$, then we have $\sum_{j=k+1}^{m-1} \hat{\beta}_t(j) \geq j_1 - k - 1$;
if $j_2 > k$, $\sum_{j=k+1}^{m-1} \hat{\beta}_t(j) \geq j_2 - k - 1$;
and if both $j_1 > k$ and $j_2>k$, then
$\sum_{j=k+1}^{m-1} \hat{\beta}_t(j) \geq j_1 + j_2 - 2k - 2$.
We thus have
\[
\sum_{j=1}^k \hat{\beta}_t(j) \leq j_1 + j_2 -2 - \max\{j_1 - k -1, 0\} - \max\{j_2-k-1,0\}
\leq \min\{j_1 - 1, k\} + \min\{j_2 - 1, k\},
\]
on one hand.
On the other hand, we observe that
\[ 
\sum_{j=1}^k (k-j+1)\hat{\mu}_t(j) \leq \max\{k - j_1 + 1, 0\} + \max\{k - j_2 + 1, 0\}
\]
since the actual positions at which the monochromatic switches happen in
the original leader sequences~$\Lambda$ are not smaller than~$j_1$ and $j_2$, respectively.
Adding these two inequalities, we obtain
\begin{align*}
    & \sum_{j=1}^k (\hat{\beta}_t(j) + (k-j+1)\hat{\mu}_t(j)) \\
    \leq & \min\{j_1 - 1, k\} + \min\{j_2 - 1, k\} + \max\{k - j_1 + 1, 0\} + \max\{k - j_2 + 1, 0\}\\
    = & \min\{j_1 - 1, k\} + \min\{j_2 - 1, k\} + \max\{-j_1+1, -k\} + \max\{-j_2 + 1, -k\} + 2k\\
    = & \min\{j_1 - 1, k\} + \min\{j_2 - 1, k\} - \min\{j_1 - 1, k\} - \min\{j_2 - 1, k\} + 2k
    = 2k,
\end{align*}
as claimed.
\end{proof}

The upper bound on~$G_k$ shown in Theorem~\ref{thm:Gk} is obtained
by summing up the inequality in Lemma~\ref{lem:bar_sub-goal}
over all elements~$t\in\{1, \ldots, n\}$.
A detailed derivation shows:
\begin{align*}
    G_k & = \sum_{j=1}^k (g_{2,j} + (k-j+1)g_{1,j}) \\
    & \leq \sum_{t\in \lambda_0} \sum_{j=1}^k (\hat{\beta}_t(j) + (k-j+1)\hat{\mu}_t(j))
    + \sum_{t\notin\lambda_0} \sum_{j=1}^k (\hat{\beta}_t(j) + (k-j+1)\hat{\mu}_t(j))\\
    & \leq \sum_{a=1}^m \min\{2(a-1), 2k\} + \sum_{t\notin \lambda_0} 2k\\
    & = \sum_{a=1}^k 2(a-1) + 2k(n-k) = k(k-1) + 2k(n-k) = k(2n-k-1).
\end{align*}

\subsubsection{The tightness}
The bounds $k(k+1) \leq G_k \leq k(2n-k-1)$ are shown to be
tight by constructions. 

\begin{restatable}{lemma}{Gktight} \label{lem:Gk_tight}
    For any two integers~$n,m$ with $1\leq m\leq n$,
    there exist two circular sequences~$\Pi_{n,m}$ and $\Pi'_{n,m}$ of permutations of
    $n$~elements colored with $m$~colors
    such that $G_k(\Pi_{n,m}) = k(2n-k-1)$ and $G_k(\Pi'_{n,m}) = k(k+1)$
    for any~$1\leq k\leq m-1$.
\end{restatable}
\begin{proof}
    Let $m$ and $n$ be any two integers with~$2\leq m\leq n$.
    The colors of elements~$\{1, \ldots, n\}$ are assigned as follows:
    the colors of the first $m$~elements are all distinct
    such that the color of~$t$ for $1\leq t\leq m$ is~$t$,
    and the other~$n-m$ elements are colored arbitrarily.
    Initially, set $\pi_0 = (1,2, \ldots, n)$.
    The sequence~$\Pi_{n,m} = (\pi_0, \ldots, \pi_{N-1}, \pi_N = \pi_0)$ is described 
    by a sequence~$\Sigma$ of switches~$\sigma_i$ between two neighboring elements in~$\pi_i$.
    The sequence~$\Sigma$ is again divided into $n-1$~blocks,
    called \emph{turns}.
    For each element~$t\in\{1, \ldots, n\}$,
    the turn of~$1$ is empty, $\Sigma_1 = \varnothing$;
    the turn of~$2$ consists of two switches, $\Sigma_2 = (12, 21)$;
    the turn of~$3$ consists of four switches, $\Sigma_3 = (23, 13, 31, 32)$.
    In general, 
    the turn of~$t$ consists of $2(t-1)$ switches 
    \[ \Sigma_t = ((t-1)t, (t-2)t, \ldots, 2t, 1t, t1, t2 \ldots, t(t-2), t(t-1)). \]
    Then the whole sequence~$\Sigma$ of switches is
    the concatenation of these turns in order:
    \[ \Sigma = \Sigma_1, \Sigma_2, \ldots, \Sigma_n.\]
    It is obvious that the number of switches in~$\Sigma$ is exactly $N=2\binom{n}{2} = n(n-1)$
    and every pair of two elements is switched exactly twice in~$\Pi_{n,m}$.
    
    Observe that the turn of~$t$ moves~$t$ forward to the first position
    (that is, $\pi_{(t-1)^2} = (t, 1, 2, \ldots, t-1, t+1, \ldots, n)$)
    and then back to its original position,
    so after each turn we get back to the initial permutation~$\pi_0$ 
    (that is, $\pi_{t(t-1)} =\pi_0$).
    From this observation,
    the turn of~$t$ makes exactly $2(a-1)$ bichromatic switches,
    where $a$ is the color of~$t$:
    exactly two at position~$j$ for~$1\leq j\leq a-1$.
    In addition,
    the turn of~$t$ makes exactly two monochromatic switches at position~$a$
    if the color of~$t$ is~$a$ and $\lambda_0(a) \neq t$;
    if $t = \lambda_0(a)$, then no monochromatic switch happens.
    So, the contribution of the turn of~$t$ to~$G_k(\Pi_{n,m})$ is exactly
    \[ \begin{cases}
        \min\{2(a-1), 2k\} & \text{if}~t=\lambda_0(a)\\
        2k & \text{otherwise}.
    \end{cases}
    \]
    The latter case follows from $2(a-1)$ bichromatic switches and
    two monochromatic switches at position~$a$, which contributes an amount of~$2(k - a + 1)$
    to~$G_k(\Pi_{n,m})$.
    Summing this up over all~$t$ gives us
    \[ G_k(\Pi_{n,m}) = k(2n-k-1).\]
    
    The second construction~$\Pi'_{n,m} = (\pi'_0, \ldots, \pi'_{N-1}, \pi'_N =\pi'_0)$ 
    is done in a similar fashion,
    but with a different sequence~$\Sigma'$ of switches.
    Set $\pi'_0 = (1,2, \ldots, n)$,
    and we color the elements as in the above construction~$\Pi_{n,m}$.
    The sequence~$\Sigma'$ of switches consists of turns:
    For each element~$t\in\{1, \ldots, n\}$,
    the turn of~$t$ is given as
    \[ \Sigma'_t = (t(t+1), t(t+2), \ldots, t(n-1), tn, nt, (n-1)t, \ldots, (t+2)t, (t+1)t ).\]
    Then the whole sequence~$\Sigma'$ of switches is as follows:
    \[ \Sigma' = \Sigma'_n, \Sigma'_{n-1}, \ldots, \Sigma'_1.\]
    Observe that the turn of~$t$ moves~$t$ backward to the last position
    and then back to its original position,
    so after each turn we get back to the initial permutation~$\pi'_0$.
    
    In this sequence~$\Pi'_{n,m}$, 
    note that every monochromatic switch happens at position~$m$ by construction.
    Further, observe that
    the turn of~$t$ for every~$t > k$ makes no switches at position~$j$ for~$j\leq k$, 
    neither monochromatic nor bichromatic.
    For~$t \leq k$, the turn of~$t$ makes exactly $2(k-t+1)$ bichromatic switches
    at position~$j$ for $j\leq k$.
    Hence, we have
    \[ 
    G_k(\Pi'_{n,m}) = \sum_{t=1}^k 2(k-t+1) = k(k+1)
    \]
    for any~$1\leq k\leq m-1$.
    
    Note that a more careful analysis can show that
    any circular sequence~$\Pi$ determined by the sequence~$\Sigma_{n,m}$ 
    (or~$\Sigma'_{n,m}$, resp.)
    of switches with any arbitrary initial permutation~$\pi_0$
    yields the same quantity~$G_k(\Pi) = k(2n-k-1)$ ($G_k(\Pi) = k(k+1)$, resp.).
\end{proof}


%
\section{Algorithm}
\label{sec:alg}

The minimal color Voronoi diagram $\CVD_k(S)$ can be computed
iteratively from 
order~$1$ up to~$k$ in $O(k^2n + n \log n)$ expected  or  $O(k^2 n\log n)$
worst-case time~\cite{bop-hocvdccsf-25}. 
However, this does not apply  to 
the maximal diagram 
because the order~$k{-}1$ diagram $\mCVD_{k-1}(S)$ does not contain sufficient
information to determine $\mCVD_k(S)$. 
For point sites under convex distance functions the following property holds: 
\emph{the sequence of sites defining unbounded regions of the minimal and the
maximal diagram is the same}.
In that case $\mCVD_k(S)$ can be obtained iteratively by computing $\CVD_k(S)$ at the same time~\cite{bop-hocvdccsf-25}.
However, this property does not hold for 
general sites 
nor for abstract Voronoi diagrams. 
\deleted{
Bae, Oliver and Papadopoulou~\cite{bop-hocvdccsf-25} 
presented an iterative algorithm that computes concrete color Voronoi diagrams from order~$1$ up to~$k$ in $O(k^2n + n \log n)$ expected time or in $O(k^2 n\log n)$ worst-case time. It assumes that (1) they follow the model of abstract Voronoi diagrams and that (2) the sequence of sites defining unbounded regions of the minimal diagram is the same as that of the maximal diagram. More precisely, on one hand, the second assumption (2) is not required for computing the minimal diagrams, so their algorithm already works for computing the minimal abstract color Voronoi diagrams. On the other hand, it is already described in \cite[Theorem~27]{bop-hocvdccsf-25} how to compute the maximal diagrams, provided the sequence of sites defining unbounded regions of maximal diagrams of each order. In the following, we present a sufficiently fast algorithm to directly compute the unbounded edges in the refined maximal diagrams $\mCVD^*_1(S), \ldots, \mCVD^*_k(S)$, and thus show that the second assumption (2) can be dropped.

This results in the first algorithm to compute the $\mCVD_{k}$ for segment
sites, and the first iterative algorithm we know of to compute the diagrams in
reverse order $\VD_{n-1}(S), \ldots, \VD_k(S)$.
}

In this paper we devise a direct divide-and-conquer
algorithm to compute the unbounded edges of $\mCVD_k(S)$, after which
$\mCVD_k(S)$ can be computed 
from $\mCVD_{k-1}(S)$.
This results in the first algorithm to compute the $\mCVD_{k}(S)$
for generalized sites, such as segments or disks,
including the order~$k$ simple polygon Voronoi diagram, and the first iterative
algorithm we know of to compute the ordinary diagrams in reverse order
$\VD_{n-1}(S), \ldots, \VD_k(S)$,
which is  efficient for large values of $k$.
Let the diagrams $\CVD^*_1(S), \ldots, \CVD^*_k(S)$ and $\mCVD^*_1(S), \ldots,
\mCVD^*_k(S)$ be abbreviated as $\CVD^*_{\leq k}(S)$ and $\mCVD^*_{\leq k}(S)$
respectively.

Our divide-and-conquer algorithm splits the set of colors $K$ at each step.
For the base case, if $k \geq m/2$, then we obtain the unbounded edges in
$\mCVD^*_{\leq k}(S)$ by walking along $\Gamma$ and computing the unbounded
edges of $\mCVD^*_{\leq m}(S)$; see Lemma~\ref{lem:alg_base_case} that details the conquer step.
If $k < m/2$,
we divide the $m$ colors in $K$ into two disjoint subsets, $K_1$ and $K_2$, of
roughly equal size; 
let $S_1$ and $S_2$ be the sites of the colors in $K_1$ and $K_2$ respectively;
see  Lemma~\ref{lem:alg_merge}.

\begin{restatable}{lemma}{algbasecase} \label{lem:alg_base_case}
	The unbounded edges in $\mCVD^*_{\leq m}(S)$ can be computed in $O(m^2(n-m+1) + n \log n)$ time.
\end{restatable}
\begin{proof}
    For each $x \in \Gamma$ that avoids any bisecting curves in $\Bisectors$, let $\pi(x)$ be the permutation of the sites induced by $\spreceq_x$. Let $\bar{\pi}(x)$ be the reverse of the permutation $\pi(x)$, and let $\bar{\lambda}(x)$ be the subsequence of $\bar{\pi}(x)$ that contains the leaders of all colors.
    Traversing $\Gamma$ corresponds to iterating through the permutations in $\overline{\Pi}(S)$; thus, $\bar{\pi}(x)$ and $\bar{\lambda}(x)$ appear in $\overline{\Pi}(S)$ and $\Lambda(\overline{\Pi}(S))$ as $\bar{\pi}_i$ and $\bar{\lambda}_i$, respectively, for some index $0 \leq i \leq N-1$. We assume without loss of generality that we traverse $\Gamma$ clockwise.
    
    First, we compute the monochromatic unbounded edges in $\mCVD^*_{\leq m}(S)$, which correspond to the unbounded edges in $\bigcup_{i \in K}\FVD(S_i)$. For each color $i$, the unbounded edges of $\FVD(S_i)$ can be computed in $O(|S_i| \log |S_i|)$ time using a divide-and-conquer algorithm similarly to the farthest segment Voronoi diagram 
    in \cite{es-flsvd-13}. 
    
    Next, we compute the bichromatic unbounded edges in $\mCVD^*_{\leq m}(S)$ by traversing $\Gamma$. We select an arbitrary starting point $x \in \Gamma$ and compute $\bar{\pi}(x)$ and $\bar{\lambda}(x)$ in $O(n \log n)$ time. 
    Let $x_m$ be the intersection of $\Gamma$ with the next monochromatic edge encountered during the traversal from $x$. For each pair of consecutive sites $p,q$ in $\bar{\lambda}(x)$, we compute the two intersections between $\Bisector{p}{q}$ and $\Gamma$. Among these $2(m-1)$ resulting intersections, let $x_b$ be the first one we encounter traversing $\Gamma$ starting from $x$.
    
    If $x_b$ precedes $x_m$ while traversing $\Gamma$ from $x$, then $x_b$ corresponds to the next switch between the leaders of two colors in $\Lambda(\overline{\Pi}(S))$. Thus, it follows from the proof of Lemma~\ref{lem:U-G} that the unbounded edge associated with $x_b$ is the next bichromatic edge along the traversal. In this case, we record the edge associated with $x_b$ as the next unbounded edge, advance $x$ to a point on $\Gamma$ immediately following $x_b$, and update $\bar{\lambda}(x)$ by swapping $p$ and $q$. Otherwise, if $x_m$ precedes $x_b$, we advance $x$ to a point on $\Gamma$ immediately following $x_m$ and update $\bar{\lambda}(x)$ by replacing the corresponding leader. We repeat this procedure until $\Gamma$ is entirely traversed.
    
    We conclude by analyzing the time complexity. Computing 
    the monochromatic unbounded edges takes $O(n \log n)$ time. The permutations $\bar{\pi}(x)$ and $\bar{\lambda}(x)$ for the starting point are also computed in $O(n \log n)$ time. During the traversal of $\Gamma$, we compute $2(m-1)$ intersections for each unbounded edge in $\mCVD^*_{\leq m}(S)$. Since there are $O(m(n-m+1))$ unbounded edges in $\mCVD^*_{\leq m}(S)$ by Lemmas~\ref{lem:sumU} and~\ref{lem:overlay_CVDk_star}, the claimed time complexity follows.
\end{proof}


\begin{restatable}{lemma}{algmerge} \label{lem:alg_merge}
  Given the unbounded edges in $\mCVD^*_{\leq k}(S_1)$ and $\mCVD^*_{\leq k}(S_2)$, where $S_1 \cup S_2 = S$, the unbounded edges in $\mCVD^*_{\leq k}(S)$ can be computed in $O(k^2 (n-k+1))$ time. 
\end{restatable}
\begin{proof}
    Consider the overlay of the diagrams in $\mCVD^*_{\leq k}(S_1)$ and $\mCVD^*_{\leq k}(S_2)$. For each $x \in \Gamma$ that avoids any bisecting curves in $\Bisectors$, let $\pi(x, S)$ be the permutation of $S$ induced by $\spreceq_x$, let $\bar{\pi}(x, S)$ be the reverse of the permutation $\pi(x, S)$, and let $\bar{\lambda}(x, S, k)$ be the subsequence of $\bar{\pi}(x, S)$ containing the leaders of the $k$ farthest colors in $K$ with respect to $\mcolpreceq_x$. We define the permutations for the sets $S_1$ and $S_2$ respectively. Traversing $\Gamma$ corresponds to iterating through the permutations in $\overline{\Pi}(S)$, $\overline{\Pi}(S_1)$, and $\overline{\Pi}(S_2)$. We assume without loss of generality that we traverse $\Gamma$ clockwise.
    
    For a point $x \in \Gamma$, the subsequence $\bar{\lambda}(x, S, k)$ can be computed by merging $\bar{\lambda}(x, S_1, k)$ and $\bar{\lambda}(x, S_2, k)$ according to $\spreceq_x$ and selecting the $k$ farthest sites from $x$. This follows from the fact that the leaders of the $k$ farthest colors in $S$, with respect to $\spreceq_x$, must be among the leaders of the $k$ farthest colors in either $S_1$ or $S_2$.
    
    We compute the unbounded edges in $\mCVD^*_{\leq k}(S)$ by traversing $\Gamma$, as in Lemma~\ref{lem:alg_base_case}. After selecting an arbitrary starting point $x \in \Gamma$, we compute $\bar{\lambda}(x, S_1, k)$ and $\bar{\lambda}(x, S_2, k)$ by identifying the unbounded faces of $\mCVD^*_{j}(S_1)$ and $\mCVD^*_{j}(S_2)$ containing $x$, for all $1 \leq j \leq k$.
    Then, we compute $\bar{\lambda}(x, S, k)$ by merging $\bar{\lambda}(x, S_1, k)$ and $\bar{\lambda}(x, S_2, k)$ according to $\spreceq_x$. Let $x_e$ be the intersection of $\Gamma$ with the next unbounded edge from the overlay of $\mCVD^*_{\leq k}(S_1)$ and $\mCVD^*_{\leq k}(S_2)$ that is encountered during the traversal from $x$. For each pair of consecutive sites $p, q$ in $\bar{\lambda}(x, S, k)$, we compute the two intersections between $\Bisector{p}{q}$ and $\Gamma$. Let $x_b$ be the first of these $2(k-1)$ intersections we encounter along $\Gamma$ starting from $x$.
    
    If $x_b$ precedes $x_e$ while traversing $\Gamma$ from $x$, then $x_b$ corresponds to the next switch between the leaders of two colors in $\Lambda(\overline{\Pi}(S))$. As shown in the proof of Lemma~\ref{lem:U-G}, the unbounded edge associated with $x_b$ is the next bichromatic edge along the traversal. In this case, we record the edge associated with $x_b$ as the next unbounded edge, advance $x$ to a point on $\Gamma$ immediately following $x_b$, and update $\bar{\lambda}(x, S, k)$ by swapping $p$ and $q$.
    Otherwise, if $x_e$ precedes $x_b$, we advance $x$ to a point on $\Gamma$ immediately following $x_e$. We then update the corresponding subsequence (either $\bar{\lambda}(x, S_1, k)$ or $\bar{\lambda}(x, S_2, k)$) based on the type of edge associated with $x_e$: if $x_e$ is monochromatic, we replace the corresponding leader; if $x_e$ is bichromatic, we swap the appropriate pair of sites. After this update, we recompute $\bar{\lambda}(x, S, k)$ by merging the updated $\bar{\lambda}(x, S_1, k)$ and $\bar{\lambda}(x, S_2, k)$. We repeat this procedure until $\Gamma$ is entirely traversed.
    
    We conclude by analyzing the time complexity. For all $1 \leq j \leq k$, scanning the unbounded faces of $\mCVD^*_{j}(S_1)$ and $\mCVD^*_{j}(S_2)$ to locate the starting point $x$ takes $O(j(n-j+1))$ time; thus computing the initial $\bar{\lambda}(x, S_1, k)$ and $\bar{\lambda}(x, S_2, k)$ takes $O(k^2(n-k+1))$ time.
    During the traversal of $\Gamma$, we compute $2(k-1)$ intersections, and potentially merge $\bar{\lambda}(x, S_1, k)$ and $\bar{\lambda}(x, S_2, k)$ in $O(k)$ time, for each unbounded edge in $\mCVD^*_{\leq k}(S_1)$, $\mCVD^*_{\leq k}(S_2)$, and $\mCVD^*_{\leq k}(S)$. Since there are $O(k(n-k+1))$ such unbounded edges by Lemma~\ref{lem:sumU}, the claimed time complexity follows.
\end{proof}

Using the conquer step in Lemma~\ref{lem:alg_base_case} and the divide step in
Lemma~\ref{lem:alg_merge}, we 
obtain the divide-and-conquer algorithm.

\begin{restatable}{lemma}{algunboundededges} \label{lem:alg_unbounded_edges}
    For $1 \leq k \leq m$, the unbounded edges in $\mCVD^*_{\leq k}(S)$ can be computed in $O(k^2 n \log m + n \log n)$ time.
\end{restatable}
\begin{proof}
    Partition the set of colors $K$ into sets $K_1$ and $K_2$, of roughly equal sizes $|K_1| = \lfloor m/2 \rfloor$ and $|K_2| = \lceil m/2 \rceil$. The set $S$ is partitioned into $S_1$ and $S_2$, of sizes $|S_1| = n_1$ and $|S_2| = n_2$. Using Lemma~\ref{lem:alg_base_case} to compute the base case and Lemma~\ref{lem:alg_merge}
    for the merge step, the time complexity of the divide-and-conquer algorithm is given by the following recurrence.
    
    \[ T(m,n) = \begin{cases} 
        T(\lfloor m/2 \rfloor, n_1) + T(\lceil m/2 \rceil, n_2) + O(k^2 (n-k+1)) & \text{if } k < m/2 \\
        O(m^2 (n-m+1) + n \log n) & \text{otherwise}
    \end{cases} \]
    
    The number of levels in the recursion tree is $L = \log_2(m/k) + O(1)$. Let $n_{i,j}$ denote the number of sites in the $j$-th subproblem at level $i$. 
    For any level $0 \leq i \leq L$, it holds that $\sum_{j=1}^{2^i} n_{i,j} = n$.
    For a constant $c>0$, the merge step at level $i$ is bounded by:
    
    \[ \sum_{j=1}^{2^i} c k^2 (n_{i,j} - k + 1) = c k^2 n - c k^2 (k-1) 2^i \]
    
    Summing this over all levels, the total merge work across the tree is bounded by:
    
    \[ \sum_{i=0}^{L-1} \left( c k^2 n - c k^2 (k-1) 2^i \right) = c k^2 n L - c k^2 (k-1) (2^L - 1) \]
    
    Since $L = \Theta(\log(m/k + 1))$ and the number of leaves is $2^L = \Theta(m/k)$, the negative term is bounded by $\Theta(m k^2)$. Because $k \leq m \leq n$, this term is strictly dominated by $c k^2 n L = O(k^2 n \log m)$; thus the total merge cost is $O(k^2 n \log m)$.
    
    Let $m_{L,j} < 2k$ denote the number of colors in the $j$-th leaf of the recursion tree.
    For a constant $c'>0$, the complexity of the base cases is bounded by:
    
    \[ \sum_{j=1}^{2^L} c' (m_{L,j}^2 (n_{L,j} - m_{L,j} + 1) + n_{L,j} \log n_{L,j}) \]
    
    Because $m_{L,j} < 2k$, we can bound $m_{L,j}^2$ by $O(k^2)$. Summing across all leaves yields $O(k^2 n + n \log n)$. Adding the total merge and base case costs together results in the claimed complexity.
\end{proof}

Using Lemma~\ref{lem:alg_unbounded_edges} as a subroutine to the iterative
algorithm \cite[Theorem~27]{bop-hocvdccsf-25}, we can compute $\CVD^*_k(S)$
for abstract Voronoi diagrams 
without any additional assumption.

\begin{theorem} \label{thm:alg_iterative_cvdk}
  Let $S$ be a set of $n$~colored sites with $m\leq n$~colors.
  Given an admissible system~$\Bisectors$ of bisecting curves for~$S$ and an
  integer~$k$ with~$1\leq k\leq m-1$, we can compute $\CVD^*_1(S), \ldots,
  \CVD^*_k(S)$ and $\mCVD^*_1(S), \ldots, \mCVD^*_k(S)$ in $O(k^2n \log n)$ time.
  $\CVD^*_1(S), \ldots, \CVD^*_k(S)$ can also be computed in expected $O(k^2n+n\log n)$ time.
\end{theorem}

Recall that if $n=m$, then $\mCVD_{k}(S) = \VD_{n-k}(S)$. Thus, using Theorem~\ref{thm:alg_iterative_cvdk} we can iteratively compute higher-order Voronoi diagrams starting from the $\FVD(S)$. 

\begin{corollary}
  Let $S$ be a set of $n$~uncolored sites.
  Given an admissible system~$\Bisectors$ of bisecting curves for~$S$ and an
  integer~$k$ with~$1\leq k\leq n-1$, we can compute $\VD_{n-1}(S), \ldots,
  \VD_{k}(S)$ in $O((n-k)^2n \log n)$ time. 
\end{corollary}

\section{Concluding remarks}
\label{sec:conclusion}


We have proved the exact maximum number~$4k(n-k)-2n$
of vertices in the order-$k$ abstract color Voronoi diagrams
$\CVD_k(S)$ and $\mCVD_k(S)$.
The main ingredients of our proof are
the colorful Clarkson--Shor technique~\cite{bop-hocvdccsf-25}
and new  tight bounds on circular sequences of permutations of colored elements.
The notion of permutations of colored elements is of independent interest
as a purely 
combinatorial concept,
which might find more applications in analyzing abstract or concrete
geometric structures
defined by a set of colored objects.

As a concrete case of abstract color Voronoi diagrams,
we showed that the complexity of the order-$k$ polygon Voronoi diagram
for $m$ disjoint simple polygons with $n$ total vertices is
$O(\min\{k(n-k), (m-k)^2n\}$.
Prior to our result, there was no known reasonable upper bound 
except for the case~$k=m-1$, that is, the farthest polygon Voronoi
diagram~\cite{cegghlln-fpvd-11}.
Nonetheless, we do not believe that this bound is tight for
\emph{every} $1\leq k\leq m-1$.
It would be interesting to obtain tight bounds for the 
complexity of the
order-$k$ polygon Voronoi diagram.
More specifically, how many vertices can there be when $k=m/2$?

Our iterative algorithms to 
construct abstract color Voronoi diagrams
of orders~$1$ to~$k$ in $O(k^2 n\log n)$ time
compares to the algorithm of~\cite{bop-hocvdccsf-25},
and even to the first algorithm by Lee~\cite{l-knnvdp-82} that computes
Euclidean Voronoi diagrams.
For constructing the Euclidean order-$k$ diagram, there was a
recent breakthrough,
achieving the optimal time $O(kn + n\log n)$~\cite{ccz-oahovdp:usn-24}.
Is it possible to apply such advanced algorithmic techniques for
faster construction
of order-$k$ color Voronoi diagrams, either abstract or concrete?
For abstract order-$k$ diagrams of non-colored sites, 
it was shown that a randomized incremental construction
can be implemented in $O(k(n-k)\log^2 n + n\log^3 n)$ expected
time~\cite{bkl-erahoavd-19}.

{
    \bibliography{references}
}

\end{document}